\documentclass[11pt, letterpaper, fleqn]{article}
\pdfoutput=1
\usepackage{amsmath}
\usepackage{amsthm}
\usepackage{latexsym}
\usepackage{xspace}
\usepackage{amssymb}
\usepackage{fancybox}
\usepackage{epsfig, url}
\usepackage{fullpage, times}
\usepackage{mathtools, wasysym}
\usepackage{cite}
\usepackage{thmtools}
\usepackage{thm-restate}
\usepackage{enumerate}

\newtheorem{theorem}{Theorem}

\newtheorem{definition}[theorem]{Definition}
\newtheorem{lemma}[theorem]{Lemma}
\newtheorem{corollary}[theorem]{Corollary}
\newtheorem{fact}[theorem]{Fact}
\newtheorem{proposition}[theorem]{Proposition}

\theoremstyle{plain}

\newtheorem{problem}{Problem}

 {
    \begin{enumerate}}{\end{enumerate}}

\newcommand{\genericdomain}{\mathcal{X}}
\newcommand{\domain}{\mathcal{D}}
\newcommand{\reals}{\mathbb{R}}
\newcommand{\R}{\mathbb{R}}

\newcommand{\sgn}{\mathrm{sign}}

\newcommand{\poly}{\mathrm{poly}}

\newcommand{\polylog}{\mathrm{polylog}}

\newcommand{\ED}{\mathsf{ ED}}

\newcommand{\SURJ}{\mathsf{ SURJ}}

\newcommand{\OR}{\mathsf{ OR}}

\newcommand{\MAJ}{\mathsf{MAJ}}

\newcommand{\bits}{\{-1,1\}}

\newcommand{\promcomp}{G^{\le N}}
\newcommand{\symcomp}{G^{\operatorname{prop}}}
\newcommand{\symsymcomp}{\tilde{G}^{\operatorname{prop}}}
\newcommand{\ls}{\star}

\DeclareMathAlphabet{\mathpzc}{OT1}{pzc}{m}{it}
\DeclareMathAlphabet{\mathcal}{OMS}{cmsy}{m}{n}

\newcommand{\eat}[1]{}

\newcommand{\N}{\mathbb{N}}

\newcommand{\Ind}{\mathbb{I}}
\newcommand{\eps}{\varepsilon}

\usepackage{color}

\newcommand{\ignore}[1]{}
\newcommand{\provisionallyremove}[1]{}

\newcommand{\adeg}{\widetilde{\operatorname{deg}}}

\newcommand{\AND}{\mathsf{AND}}

\renewcommand{\sgn}{\operatorname{sgn}}

\makeatletter
\newcommand*{\itemequation}[3][]{%
  \item
  \begingroup
    \refstepcounter{equation}%
    \ifx\\#1\\%
    \else
      \label{#1}%
    \fi
    \sbox0{#2}%
    \sbox2{$\displaystyle#3\m@th$}%
    \sbox4{ \@eqnnum}%
    \dimen@=.5\dimexpr\linewidth-\wd2\relax
    \let\CenterInSpace=N%
    \ifcase
        \ifdim\wd0>\dimen@
          \z@
        \else
          \ifdim\wd4>\dimen@
            \z@
          \else
            \@ne
          \fi
        \fi
      \let\CenterInSpace=Y%
    \fi
    \ifdim\dimexpr\wd0+\wd2+\wd4\relax>\linewidth
      \@latex@warning{Equation is too large}%
    \fi
    \noindent
    \rlap{\copy0}%
    \ifx\CenterInSpace Y%
      \rlap{\hbox to \linewidth{\kern\wd0\hss\copy2\hss\kern\wd4}}%
    \else
      \rlap{\hbox to \linewidth{\hfill\copy2\hfill}}%
    \fi
    \hbox to \linewidth{\hfill\copy4}%
    \hspace{0pt}
  \endgroup
  \ignorespaces
}
\makeatother

\title{A Nearly Optimal Lower Bound on the Approximate Degree of AC$^0$}
\author{Mark Bun\thanks{Princeton University.}\\ \texttt{mbun@cs.princeton.edu}
   \and 
Justin Thaler\thanks{Georgetown University.}   \\ \texttt{justin.thaler@georgetown.edu}
}

\date{}
\begin{document}
\maketitle

\begin{abstract}
The approximate degree of a Boolean function $f \colon \{-1, 1\}^n \rightarrow \{-1, 1\}$ 
is the least degree of a real polynomial that approximates $f$ pointwise to error at most $1/3$. We introduce a generic method for increasing the approximate degree of a given function, while preserving its computability by constant-depth circuits.

Specifically, we show how to transform any Boolean function $f$ with approximate degree 
$d$ into a function $F$ on $O(n \cdot \polylog(n))$ variables
with approximate degree at least $D = \Omega(n^{1/3} \cdot d^{2/3})$.
In particular, if $d= n^{1-\Omega(1)}$, then $D$ is polynomially larger than $d$.
Moreover, if $f$ is computed by a polynomial-size Boolean circuit of constant depth, then so is $F$.

By recursively applying our transformation, for any constant $\delta > 0 $ we exhibit  an AC$^0$ function of approximate degree $\Omega(n^{1-\delta})$. This improves over the best previous lower bound of $\Omega(n^{2/3})$ due to Aaronson and Shi (J. ACM 2004), and nearly matches the trivial upper bound of $n$ that holds for any function. Our lower bounds also apply to (quasipolynomial-size) DNFs of polylogarithmic width.

We describe several applications of these results. We give:
\begin{itemize}
\item For any constant $\delta > 0$, an $\Omega(n^{1-\delta})$ lower bound on the quantum communication complexity of a function in AC$^0$.
\item A Boolean function $f$ with approximate degree at least $C(f)^{2-o(1)}$,
where $C(f)$ is the \emph{certificate complexity} of $f$. This separation is optimal up to the $o(1)$ term in the exponent.
\item Improved secret sharing schemes with reconstruction procedures in AC$^0$.

\end{itemize}
\end{abstract}

\section{Introduction}
The $\eps$-approximate degree of a Boolean function $f \colon \{-1, 1\}^n \rightarrow \{-1, 1\}$,
denoted $\adeg_{\eps}(f)$, 
is the least degree of a real polynomial that approximates $f$ pointwise to error at most $\eps$.
By convention, $\adeg(f)$ is used to denote $\adeg_{1/3}(f)$, and this quantity is referred
to without qualification as the \emph{approximate degree} of $f$. 
The choice of the constant $1/3$ is arbitrary, as
$\adeg(f)$ is related to $\adeg_{\eps}(f)$ by a constant factor for any constant $\eps \in (0, 1)$. Any Boolean function $f$ has an exact representation as a multilinear polynomial of degree at most $n$, so the approximate degree of $f$ is always at most $n$.

Approximate degree is a natural measure of the complexity of a Boolean function,
with a wide variety of applications throughout theoretical computer science. 
For example, upper bounds on approximate degree underly many state-of-the-art learning 
algorithms \cite{ksdnf, klivansservedioomb, agnostic, servediotanthaler,
readonceformulae, colt, osnewbounds}, algorithmic approximations for the inclusion-exclusion principle~\cite{kahn,sherstovinclusion}, and algorithms for differentially private data release
\cite{difpriv1, difpriv2}. Very recently, approximate degree upper bounds have also been used to show new complexity-theoretic \emph{lower bounds}.
In particular, upper bounds on the approximate degree of Boolean formulae underly
the best known lower bounds
on the formula complexity and graph complexity of explicit functions \cite{tal1, tal2, tal3}.

Meanwhile, lower bounds on approximate degree
have enabled significant progress in quantum query complexity~\cite{qqc1, qqc2, aaronsonshi}, communication complexity~\cite{patmat, bvdw, comm1, comm2, comm3, comm4, comm6, comm7, comm8, sherstovsurvey}, circuit complexity~\cite{mp, sherstovmajmaj},
oracle separations~\cite{beigel, bchtv}, and secret-sharing~\cite{viola}.
In particular, approximate degree has been established as one of the most promising tools available 
for understanding the complexity of constant-depth Boolean circuits\footnote{In this paper, all circuits are Boolean and
of polynomial size unless otherwise specified.} (captured
by the complexity class AC$^0$).
Indeed, approximate degree lower bounds lie at the heart of the best known
bounds on the complexity of AC$^0$ under measures such as
sign-rank, discrepancy and margin complexity,
Majority-of-Threshold and Threshold-of-Majority circuit size, and more. 

Despite all of these applications, progress in understanding
approximate degree has been
slow and difficult. 
As noted by many authors,
the following basic problem remains unresolved~\cite{bt14, btdl, bt16, bttoc, beame, sherstov15, sherstov14, viola}. 

\begin{problem} \label{problem:bounded}
Is there a constant-depth circuit in $n$
variables with approximate degree $\Omega(n)$?
\end{problem}

Prior to this work, the best result in this direction was 
Aaronson and Shi's well-known $\Omega(n^{2/3})$ lower bound
on the approximate degree of the Element Distinctness function ($\ED$ for short). 
In this paper, we nearly resolve Open Problem~\ref{problem:bounded}. Specifically, for any constant $\delta > 0$, we
exhibit
an explicit constant-depth circuit $\mathcal{C}$ with approximate degree $\Omega(n^{1-\delta})$. 
Moreover, the circuit $\mathcal{C}$ that we exhibit has depth $O(\log(1/\delta))$.
Our lower bound
also applies to DNF formulae of polylogarithmic width (and quasipolynomial size).

 \medskip \noindent \textbf{Applications.}
We describe several consequences of the above results in complexity theory and cryptography. (Nevertheless, the list of applications we state here is not exhaustive.) We state these results somewhat informally in this introduction,
 leaving details to Section \ref{sec:apps}. Specifically:
 \begin{itemize}
 \item  For any constant $\delta > 0$, we obtain an $\Omega(n^{1-\delta})$ lower bounds on the quantum communication complexity of AC$^0$. This nearly matches the trivial $O(n)$ upper bound that holds for any function.
   \item We exhibit a function $f$ with approximate degree at least $C(f)^{2-o(1)}$,
where $C(f)$ is the \emph{certificate complexity} of $f$. This separation is optimal up to the $o(1)$ term in the exponent. 
 The previous best result was a power-7/6 separation, reported by Aaronson et al. \cite{cheatsheets}. 
   \item We give improved secret sharing schemes with reconstruction procedures in AC$^0$.
  \end{itemize}

While the first and third applications follow by combining our approximate degree lower bounds
 with prior works in a black box manner \cite{patmat, viola}, the second application
requires some additional effort.

\subsection{Prior Work on Approximate Degree}
\subsubsection{Early Results via Symmetrization}
The notion of approximate degree was introduced in seminal work of
Nisan and Szegedy \cite{nisanszegedy}, who proved a tight $\Omega(n^{1/2})$
lower bound on the approximate degree of the functions $\OR_n$ and $\AND_n$.\footnote{Whenever
it is not clear from context, we use subscripts to denote the number of variables on which a function is defined.}
Nisan and Szegedy's proof exploited a powerful technique known as
\emph{symmetrization}, which was introduced in the late 1960's by Minsky and Papert \cite{mp}.
Until recently, symmetrization was the primary tool available for proving approximate degree lower bounds \cite{ambainis, paturi, aaronsonshi,
beigel, servediotanthaler, podolskii}.

Symmetrization arguments proceed in two steps. First, a polynomial 
$p$ on $n$ variables (which is assumed to approximate
the target function $f$) is transformed into a univariate polynomial $q$ in such a way that
$\deg(q) \leq \deg(p)$. Second, a lower bound on $\deg(q)$ is proved, using techniques 
tailored to the analysis of univariate polynomials.

Although powerful,
symmetrization is inherently lossy: by turning a polynomial $p$ on $n$ variables into a univariate polynomial $q$, 
information about $p$ is necessarily thrown away. 
Hence, several works identified the development of non-symmetrization techniques for lower bounding the approximate degree of Boolean functions as an important research direction (e.g., \cite{focstutorial, sherstovhalfspaces1, sherstovsurvey}).
A relatively new such lower-bound technique called the \emph{method of dual polynomials} plays an essential role in our paper.

\subsection{The Method of Dual Polynomials and the AND-OR Tree}

A dual polynomial is a dual solution to a certain linear program capturing the approximate degree of
any function. These polynomials act as certificates of the high approximate degree of a function.
Strong LP duality implies that the technique is lossless, in contrast to symmetrization. That is, for any
function $f$ and any $\eps$, there is always some dual polynomial $\psi$ that witnesses a tight $\eps$-approximate degree
lower bound for $f$. 

A dual polynomial that witnesses the fact that $\adeg_{\eps}(f_n) \geq d$ is a 
function $\psi \colon \bits^n \to \bits$ satisfying three properties:
\begin{itemize}
\item $\sum_{x \in \bits^n}\psi(x) \cdot f(x) > \eps$. If $\psi$ satisfies this condition, it is said to be \emph{well-correlated} with $f$.
\item  $\sum_{x \in \bits^n} |\psi(x)| = 1$. If $\psi$ satisfies this condition, it is said to have $\ell_1$-norm equal to 1.
\item For all polynomials $p \colon \bits^n \to \R$ of degree less than $d$, we have $\sum_{x \in \bits^n} p(x) \cdot \psi(x) = 0$. 
If $\psi$ satisfies this condition, it is said to have \emph{pure high degree} at least $d$.
\end{itemize}

One success story for the method of dual polynomials is the resolution of the approximate degree of the two-level AND-OR tree. For many years, this was the simplest function whose approximate degree resisted characterization by symmetrization methods~ \cite{nisanszegedy, shi, ambainis, sherstovhalfspaces1}. Given two functions $f_M, g_N$,  let $f \circ g \colon \bits^{M \cdot N} \to \bits$ denote their \emph{block composition}, i.e.,
$f \circ g = f(g, \dots, g)$.

\begin{theorem} \label{theorem:and-or}
The approximate degree of the function $\AND_M \circ \OR_N$ is $\Theta(\sqrt{M \cdot N})$.
\end{theorem}

Ideas pertaining to both the upper and lower bounds of Theorem~\ref{theorem:and-or} will be useful to understanding the results in this paper. The upper bound of Theorem~\ref{theorem:and-or} was established by
H{\o}yer, Mosca, and de Wolf \cite{hoyer}, who designed a quantum query algorithm to prove that
$\adeg(\AND_M \circ \OR_N) = O(\sqrt{MN})$. Later, Sherstov \cite{sherstovrobust}
proved the following more general result.

\begin{theorem}[Sherstov \cite{sherstovrobust}] For \emph{any}
Boolean functions $f, g$, we have $\adeg(f \circ g) = O(\adeg(f) \cdot \adeg(g))$. 
\label{thm:robust}
\end{theorem}
Sherstov's remarkable proof of Theorem \ref{thm:robust} is via a technique we call \emph{robustification}.
This approximation technique will be an important source of intuition for our new results.

\paragraph{Robustification.} \label{sec:robust} Sherstov \cite{sherstovrobust} showed that for any polynomial
$p \colon \bits^M \to \bits$, and every $\delta > 0$, there is a polynomial 
$p_{\text{robust}}$ of degree $O(\deg(p) +
\log(1/\delta))$ that is robust to noise in the sense that $|p(y) - p_{\text{robust}}(y + \textbf{e})| < \delta$ for all 
$y \in \bits^M$, and $\mathbf{e} \in [-1/3, 1/3]^M$.
Hence, given functions $f_M, g_N$, one can obtain an $(\eps +\delta)$-approximating polynomial
for the block composition $f_M \circ g_N$ as follows. Let $p$ be an $\eps$-approximating polynomial for $f_M$, and $q$
a $(1/3)$-approximating polynomial for $g_N$. Then the block composition 
$p^*:= p_{\text{robust}}(q, \dots, q)$ is an $(\eps + \delta)$-
approximating polynomial for $f_M \circ g_N$. Notice that the degree of $p^*$
is at most the product of the degrees
of $p_{\text{robust}}$ and $q$.

\medskip

Sherstov \cite{sherstovandor} and the authors \cite{bt13} independently
used the method of dual polynomials to obtain the matching $\Omega(\sqrt{M \cdot N})$ lower bound of Theorem~\ref{theorem:and-or}. These lower bound proofs work by
constructing (explicitly in \cite{bt13} and implicitly in \cite{sherstovandor}) an optimal dual polynomial $\psi_{\text{AND-OR}}$ for the AND-OR tree. 
Specifically, $\psi_{\text{AND-OR}}$ is obtained by taking dual polynomials $\psi_{\text{AND}}, \psi_{\text{OR}}$
respectively witnessing the fact that $\adeg(\AND_M) = \Omega(\sqrt{M})$ and $\adeg(\OR_N) = \Omega(\sqrt{N})$,
and combining them in a precise manner. 

For arbitrary Boolean functions $f$ and $g$, this method of combining dual polynomials $\psi_f$ and $\psi_g$
to obtain a dual polynomial $\psi_f \ls \psi_g$ for $f \circ g$ was introduced in earlier
line of work by Shi and Zhu~\cite{shizhu}, Lee \cite{lee} and Sherstov \cite{sherstovhalfspaces1}. Specifically, writing $x=(x_1, \dots, x_M) \in \left(\bits^N\right)^M$,
$$ (\psi_f \ls \psi_g)(x) := 2^M \cdot \psi_f(\dots, \sgn(\psi_g(x_i)), \dots) \cdot \prod_{i = 1}^{M} |\psi_g(x_i)|.$$
This technique of combining dual witnesses, which we call the ``dual block'' method, will also be central to this paper. The lower bound of~\cite{sherstovandor, bt13} was obtained by refining the analysis of $\psi_f \ls \psi_g$ from~\cite{sherstovhalfspaces1} in the case where $f = \AND_M$ and $g = \OR_N$. 


As argued in subsequent work of Thaler \cite[Section 1.2.4]{thaler}, the combining method
$\psi_f \ls \psi_g$ is specifically tailored to showing optimality of the polynomial
approximation  $p^*$ for $f \circ g$ obtained via robustification.
This assertion can be made precise via complementary slackness: the dual solution 
 $\psi_f \ls \psi_g$ can be shown to obey complementary slackness in an approximate (yet precise) sense with respect
to the solution to the primal linear program corresponding to $p^*$.\label{sec:compslack}

\subsubsection{Additional Prior Work} \label{sec:additionalpriorwork}
The method of dual polynomials has recently been used 
to establish a number of new lower bounds for approximate degree
\cite{osnewbounds, thaler, sherstovhalfspaces1, sherstovhalfspaces2, 
bt14, lijie1, bchtv}.
All of these results focus on block composed functions, and can
be viewed as \emph{hardness amplification} results. Specifically,
they show that the block composition $f \circ g$ is strictly harder
to approximate by low-degree polynomials (requiring either higher degree or higher error)
than either $f$ or $g$ individually. These results have enabled progress on a number of open questions regarding the complexity of AC$^0$,
as well as oracle separations involving the polynomial hierarchy and various notions of statistical zero-knowledge proofs. 

Recently,
a handful of works  have proved stronger hardness amplification results for approximate degree
by moving beyond block composed functions \cite{btdl, podolskii}. These papers use very different
techniques than the ones we introduce in this work, as they
are focused on a different form of hardness amplification for polynomial approximation (specifically, they amplify
approximation error instead of degree). 

\subsection{Our Results and Techniques}
A major technical hurdle to progress on Problem~\ref{problem:bounded} is the need to go
beyond the block composed functions that were the focus 
of prior work. Specifically, Theorem \ref{thm:robust} 
implies that the approximate degree of $f_M \circ g_N$ (viewed
as a function of the number of inputs $M \cdot N$) is \emph{never}
higher than the approximate degree of $f_M$ or $g_N$ individually (viewed
as a function of $M$ and $N$ respectively). For example, if $f_M$ and $g_N$
both have approximate degree equal to the square root of the number of inputs (i.e., 
$\adeg(f_M) = O(\sqrt{M})$ and $\adeg(g_N) = O(\sqrt{N})$),
then the block composition $f_M \circ g_N$ has the same property (i.e., $\adeg(f_M \circ g_N) = O(\sqrt{M \cdot N})$). Our results introduce an analysis of non-block-composed functions that overcomes this hurdle.

\medskip
Quantitatively, our main lower bounds for constant-depth circuits and DNFs are as follows. 
To obtain the tightest possible results for a given circuit depth, our analysis pays close attention to whether a circuit $\mathcal{C}$ is monotone ($\mathcal{C}$ is said to be monotone if it contains no $\mathsf{NOT}$ gates).

\begin{theorem} \label{thm:quant1} \label{thm:acz}
Let $k \ge 1$ be any constant integer. Then there is an (explicitly given) monotone circuit on $n \cdot \log^{4k-4}(n)$ variables of depth $2k$, with 
$\AND$ gates at the bottom, which computes a function with approximate degree $\Omega(n^{1-2^{k-1}/3^k} \cdot \log^{3 - 2^{k+2}/3^{k}}(n))$. 
\end{theorem}


For example, Theorem \ref{thm:quant1} implies a Boolean circuit of depth $6$ on $n$ variables with approximate degree $\tilde{\Omega}(n^{23/27}) = \tilde{\Omega}(n^{0.851...})$.
\medskip

 \begin{theorem} \label{thm:quant2} \label{thm:dnfs} \label{thm:dnf}
Let $k \ge 1$ be any constant integer. Then there is an (explicitly given) monotone DNF on $n \cdot \log^{4k-4}(n)$ variables of width $O\left(\log^{2k-1}(n)\right)$ (and size $2^{O(\log^{2k}(n))}$) which computes a function with approximate degree $\Omega(n^{1-2^{k-1}/3^k} \cdot \log^{3 - 2^{k+2}/3^k}(n))$.
\end{theorem}

Theorems \ref{thm:quant1} and \ref{thm:quant2} are in fact corollaries of a more general
hardness amplification theorem. This result shows how to take any Boolean function $f$ and transform it into a related function $g$
on roughly the same number of variables that has significantly higher approximate
degree (unless the approximate degree of $f$ is already $\tilde{\Omega}(n)$). Moreover, if $f$ is computed by a low-depth circuit, then $g$ is as well. 

\begin{theorem}\label{thm:main}
Let $f \colon \bits^n \to \bits$ with $\adeg(f) = d$. Then $f$ can be transformed into a related function $g \colon \{-1, 1\}^{m} \rightarrow \{-1, 1\}$ with $m=O(n\log^4 n)$ and $\adeg(g) = \Omega(n^{1/3} \cdot d^{2/3} \cdot \log n).$ Moreover, $g$ satisfies the following additional properties.
\begin{align}
\bullet \quad &\text{If $f$ is computed by a circuit of depth $k$, then $g$ is computed by a circuit of depth $k+3$.} \label{propeasy} \\
\bullet \quad &\text{If $f$ is computed by a monotone circuit of depth $k$ with 
$\AND$ gates at the bottom,} \nonumber \\ & \text {then $g$ is computed by a monotone circuit of depth $k+2$ with $\AND$ gates at the bottom.\label{propaczhard}} \\
\bullet \quad &\text{If $f$ is computed by monotone DNF of width $w$, then $g$ is computed by monotone DNF of} \nonumber \\ &\text{width $O(w \cdot \log^2 n)$.}\label{propdnfhard}
\end{align}

 \end{theorem}
 

 
 
 \label{sec:theorems}
 
 \subsubsection{Hardness Amplification Construction}
 The goal of this subsection is to convey the main ideas underlying
 the transformation of $f$ into the harder-to-approximate function $g$ in the statement of Theorem \ref{thm:main}. 
We focus on illustrating these ideas when we start with the function $f = \AND_R$, where we assume for simplicity that $R$ is a power of $2$. Let $n=N \log R$ for a parameter $N$ to be determined later.\footnote{All logarithms in this paper are taken in base $2$.} Consider the function $$\textsf{SURJECTIVITY}\colon \bits^n \to \bits$$ (\textsf{SURJ}$_{N, R}$ for short) defined as follows. \textsf{SURJ}$_{N, R}$ interprets its input $s$ as a list of $N$ numbers $(s_1, \dots, s_N)$
from a range $[R]$. The function \textsf{SURJ}$_{N, R}(s)=-1$ if and only if every element of the range $[R]$ appears at least once in the list.\footnote{As is standard, we associate $-1$ with logical TRUE and $+1$ with logical FALSE throughout.}

When we apply Theorem \ref{thm:main} to $f=\AND_R$, the harder function $g$ we construct is precisely $\SURJ_{N, R}$ (for a suitable choice of $N \leq \tilde{O}(R)$).
Before describing our transformation for general $f$, we provide some intuition for why $\SURJ_{N, R}$ is harder to approximate than $\AND_R$.

\paragraph{Getting to Know \textsf{SURJECTIVITY}.}
It is known that  $\adeg(\textsf{SURJ}_{N, R}) = \tilde{\Omega}(n^{2/3})$ when $R=N/2$ \cite{aaronsonshi}. We do not improve this lower bound for $\textsf{SURJ}_{N, R}$, but we give a much more general and intuitive proof for it.
The best known upper bound on $\adeg(\textsf{SURJ}_{N, R})$ is the trivial $O(n)$ that holds 
for any function on $n$ variables. 

\label{sec:y}
Although this upper bound is trivial, the following is an instructive way to achieve it. For $(i, j) \in [R] \times [N]$,
let\footnote{We clarify that earlier work \cite{ambainis} using similar notation reverses
the roles of $i$ and $j$ in the definition of $y_{ij}(s)$. We depart from the convention
of earlier work because it simplifies the expression of the harder function $g$ exhibited in
Theorem \ref{thm:main}.} \[y_{ij}(s) =  \begin{cases} -1 \text{ if } s_j=i \\
1 \text{ otherwise.} \end{cases} \] 
Observe that $y_{ij}(s)$ is exactly computed by a polynomial in $s$ of degree at most $\log R$,
as $y_{ij}(s)$ depends on only $\log R$ bits of $s$. For brevity, we will typically denote
$y_{ij}(s)$ by $y_{ij}$, but the reader should always bear in mind that $y_{ij}$ is a function of $s$.

\begin{figure}
\begin{center} 
\includegraphics[width=2.5in]{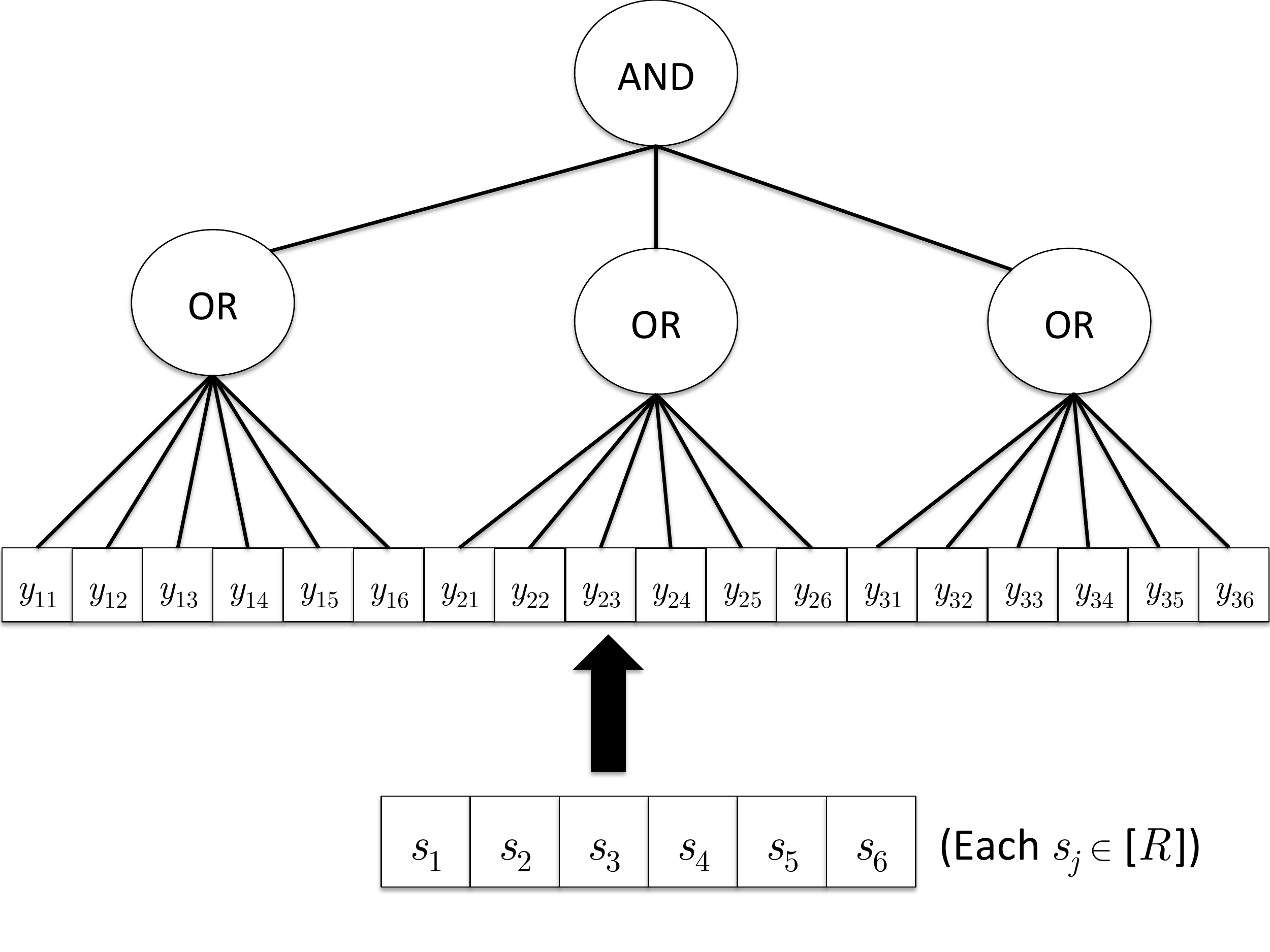} 
\mbox{                          }
\includegraphics[width=2.5in]{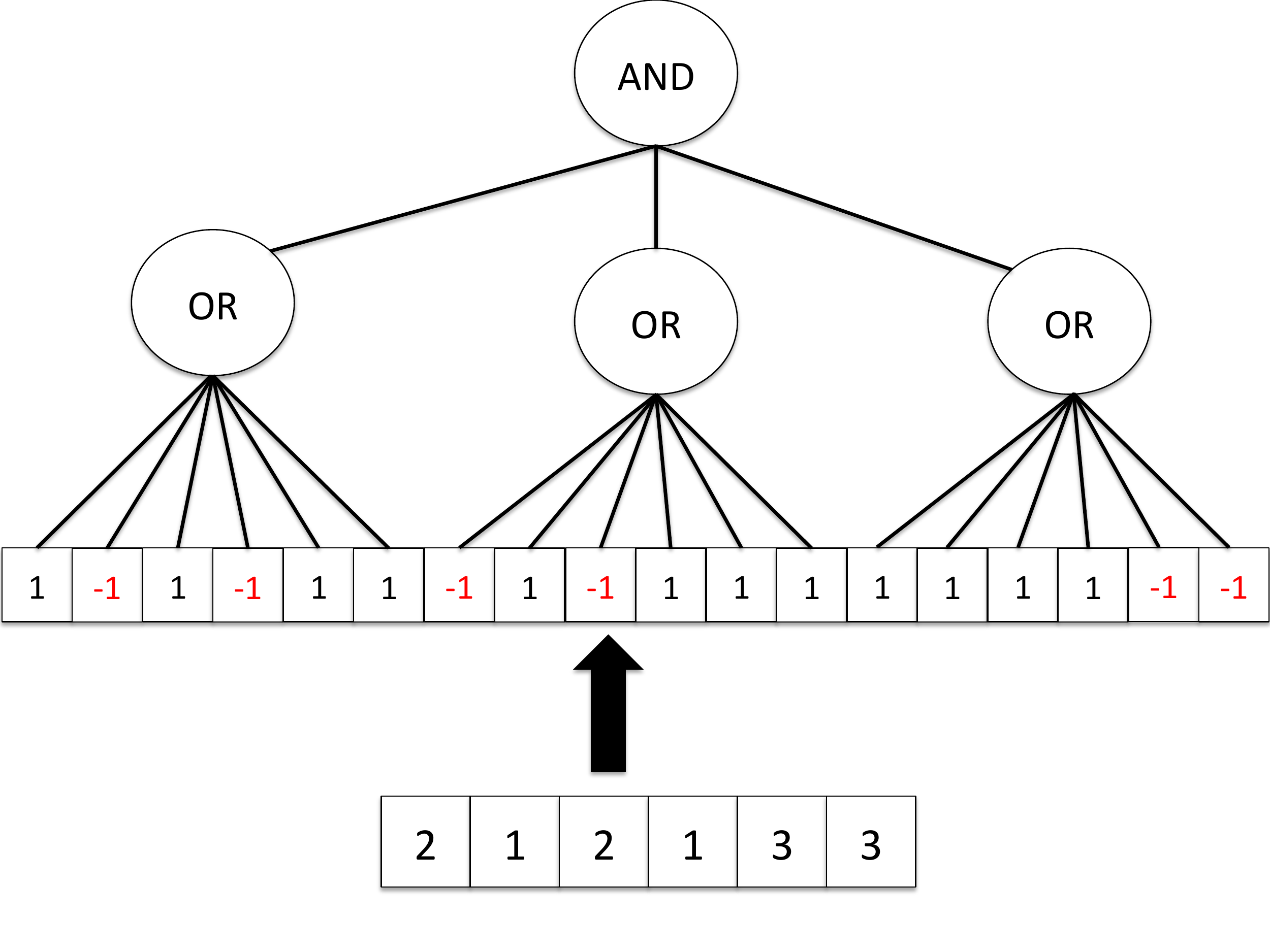}
\end{center} 
\caption{Depiction of Equation \eqref{eq:surj} when $N=6$ and $R=3$.}
\label{fig:surj}
\end{figure}

Clearly, it holds that: \begin{equation} \label{eq:surj} \textsf{SURJ}_{N, R}(s) \!=\! \AND_R(\OR_N(y_{1,1}, \dots, y_{1,N}), \dots, \OR_N(y_{R,1}, \dots, y_{R,N})).\end{equation}
Equality \eqref{eq:surj} is depicted in Figure \ref{fig:surj} in the special case $N=6, R=3$. 
Let $p^*$ be the polynomial approximation of degree $O(\sqrt{R \cdot N})$ for
the block composed function $\AND_R \circ \OR_N$ obtained via robustification (cf. Section \ref{sec:robust}). 
Then $$p^*(y_{1, 1}, \dots, y_{1, N}, \dots, y_{R, 1}, \dots, y_{R, N})$$
approximates $\textsf{SURJ}_{N, R}$, and has degree $O(\deg(p^*) \cdot \log R)$. If $N=O(R)$, then this degree
bound is $O(N \log R) = O(n)$. 

Our analysis in the proof of Theorem \ref{thm:main} is tailored to showing a sense in which this robustification-based approximation method is nearly optimal.
Unsurprisingly, our analysis makes heavy use of the dual block method of combining dual witnesses~\cite{shizhu,lee,sherstovhalfspaces1}, as this method is tailored to showing optimality of robustification-based approximations (cf. Section \ref{sec:compslack}). However, there are several technical challenges to overcome,
owing to the fact that Equation \eqref{eq:surj} does not express $\SURJ$ as a genuine block composition
(since a single bit of the input $s \in \bits^{N \cdot \log R}$ affects $R$ of the variables $y_{ij}$).

\medskip 
\noindent \textbf{The Transformation for General Functions.}
Recall from the preceding discussion that when applying
our hardness-amplifying transformation to the function $f = \AND_R$,
the harder function (on $n = N \cdot \log R$ bits, for some $N=\tilde{O}(R)$) takes the form $\SURJ_{N, R} =  \AND_R(\OR_N(y_{1,1}, \dots, y_{1,N}), \dots, \OR_N(y_{R,1} \dots, y_{R,N}))$.
This suggests that for general functions $f \colon \bits^R \to \bits$, one should consider the
transformed function 
$$F(s) := f(\OR_N(y_{1,1}, \dots, y_{1,N}), \dots, \OR_N(y_{R,1}, \dots, y_{R,N})).$$
Unfortunately, this simple candidate fails spectacularly.
Consider the particular case where $f=\OR_R$. 
It is easy to see that $$\OR_R(\OR_N(y_{1,1}, \dots, y_{1,N}), \dots, \OR_N(y_{R,1}, \dots, y_{R,N}))$$
evaluates to $-1$ on \emph{all} inputs $s \in \bits^{N \cdot \log R}$. Hence, it has (exact) degree equal to 0.

Fortunately, we are able to show that a modification of the above candidate does work for general functions $f_R$.
Let $R' = R \log R$.
Still simplifying, but only slightly,  
the harder function that we exhibit
is $g \colon \bits^{N \cdot \log(R')} \to \bits$ defined via:
$$g(s) = (f \circ \AND_{\log R})(\OR_N(y_{1,1},\! \dots,\! y_{1,N}), \!\dots, \!\OR_N(y_{R',1},\! \dots\!, y_{R',N})).$$
 \subsubsection{Hardness Amplification Analysis}
 For expository purposes, we again describe the main ideas of our analysis in the case where $f=\AND_R$.
Recall that in this case, the harder function $g$ exhibited in Theorem \ref{thm:main}
is \textsf{SURJ}$_{N, R}$ on $n=N \cdot \log R$ bits. 
Moreover,
 in order to approximate  $\textsf{SURJ}_{N, R}$, it is \textit{sufficient}
to approximate the \emph{block composed function} $\AND_R \circ \OR_N$. 
This can be done by a polynomial of degree $O(\sqrt{R \cdot N})$ using robustification. 

The goal of our analysis is to show that there is a sense in which this approximation method for \textsf{SURJ}$_{N, R}$ is essentially optimal.
Quantitatively, our analysis yields an $\Omega(R^{2/3})$ lower bound on the approximate degree of \textsf{SURJ}$_{N, R}$.

At a high level, our analysis proceeds in two stages.
In the first stage (Section~\ref{sec:step1}), we give a reduction showing that to approximate $\textsf{SURJ}_{N, R}(x)$,
it is \textit{necessary} to approximate 
$\AND_R \circ \OR_N$, 
under the promise that
the input has Hamming weight \textit{at most} $N$.
This reduction is somewhat subtle, but conceptually crucial to our results. Nevertheless, at the technical level, it is
a straightforward application of a symmetrization
argument due to Ambainis \cite{ambainis}.

In the second stage (Section~\ref{sec:step2}), we prove that approximating $\AND_R \circ \OR_N$ 
under the above promise requires degree $\Omega(R^{2/3})$. Executing this second stage is the more technically involved part of our proof, and we devote the remainder of this informal overview to it.
Specifically, for some $N=\tilde{O}(R)$, we must construct a dual polynomial $\psi_{\text{AND-OR}}$ witnessing
the fact that 
$\adeg(\AND_R \circ \OR_N) = \Omega(R^{2/3})$, such that  
$\psi_{\text{AND-OR}}$ is supported exclusively on inputs of Hamming weight at most $N$. 

As a first attempt, one could consider the dual polynomial 
$\psi_{\text{AND}} \ls \psi_{\text{OR}}$  (cf. Section \ref{sec:compslack}) used in our prior work \cite{bt13}
to lower bound the approximate degree of the AND-OR tree.
Unfortunately, this dual polynomial has inputs of Hamming weight as large as $\Omega(R \cdot N)$
in its support. 

Our strategy for handling this issue is to modify $\psi_{\text{AND}} \ls \psi_{\text{OR}}$ by post-processing it to zero out all of the mass it places on inputs of Hamming weight more than $N$. This must be done 
without significantly affecting its pure high degree, its $\ell_1$-norm, or its correlation with $\AND_R \circ \OR_N$.
In more detail, let $|y|$ denote the Hamming weight of an input $y \in \bits^{R \cdot N}$, and
suppose that we can show \begin{equation} \label{eq:introkey} \sum_{|y| >  N} |(\psi_{\text{AND}} \ls \psi_{\text{OR}})(y)| \ll R^{-D}.\end{equation}
Intuitively, if Inequality \eqref{eq:introkey} holds for a large value of $D$, then inputs 
of Hamming weight greater than $N$ are not very important to the dual witness $\psi_{\text{AND}} \ls \psi_{\text{OR}}$,
and hence it is plausible that the lower bound witnessed by $\psi_{\text{AND}} \ls \psi_{\text{OR}}$ holds even
if such inputs are ignored completely. 

To make the above intuition precise, we use a result of Razborov and Sherstov \cite{sherstovrazborov}
to establish that Inequality \eqref{eq:introkey} implies 
the existence of a (explicitly given) function $\psi_{\text{corr}} \colon \bits^{N \cdot R} \to \bits$ such that:
\begin{itemize} \item $\psi_{\text{corr}}(y)=\psi_{\text{AND}} \ls \psi_{\text{OR}}(y)$ for all $|y| > N$,
\item $\psi_{\text{corr}}$ has pure high degree $D$, and
\item $\sum_{|y| > N} |\psi_{\text{corr}}(y)| \ll R^{-D}$.
\end{itemize}

Let $\psi_{\text{AND-OR}} = C \cdot \left(\psi_{\text{AND}} \ls \psi_{\text{OR}} - \psi_{\text{corr}}\right)$, where $C \geq 1-o(1)$ is chosen so that 
 the resulting function has $\ell_1$-norm equal to $1$. Then $\psi_{\text{AND-OR}}$ has:
\begin{enumerate}
\item Pure high degree $\min\{D, \sqrt{R \cdot N}\}$,
\item The same correlation, up to a factor of $1-o(1)$, as $\psi_{\text{AND}} \ls \psi_{\text{OR}}$ has with $\AND_R \circ \OR_N$, and
\item Support restricted to inputs of Hamming weight at most $N$.
\end{enumerate}

Hence, Step 2 of the proof is complete if we can show that Inequality \eqref{eq:introkey} holds
for $D=\Omega(R^{2/3})$. 
Unfortunately, Inequality \eqref{eq:introkey} does \emph{not} hold unless we modify the dual witness $\psi_{\text{OR}}$ to satisfy
additional properties. First,
we modify $\psi_{\text{OR}}$ so that 
\begin{equation} \label{prop1} \psi_{\text{OR}} \text{ is supported
only on inputs of Hamming weight at most } R^{1/3}. \end{equation} 
Moreover, we further ensure that $\psi_{\text{OR}}$ 
is biased toward inputs of low Hamming weight in the sense
that \begin{equation} \label{prop2} \text{ For all } t \ge 0\text{, }  \sum_{|x| = t} |\psi_{\text{OR}}(x)| \lesssim 1/(t+1)^2. \end{equation} 
We can guarantee that both Conditions \eqref{prop1} and \eqref{prop2} hold
while still ensuring that $\psi_{\text{OR}}$ has pure high degree $\Omega(R^{1/6})$, as well as the same $\ell_1$-norm and correlation with $\OR_N$. (The fact that this modified dual polynomial $\psi_{\text{OR}}$ has pure high degree $\Omega(R^{1/6})$
rather than $\Omega(R^{1/2})$ is the reason we are only able to establish an $\Omega(R^{2/3})$
lower bound on the approximate degree of $\SURJ_{N, R}$, rather than $\Omega(R)$.)

We now explain why these modifications imply that 
Inequality \eqref{eq:introkey} holds
for $D=\Omega(R^{2/3})$.
Recall that $$(\psi_{\text{AND}} \ls \psi_{\text{OR}})(y_1, \dots, y_R) = 2^R \cdot \psi_{\text{AND}}(\dots, \sgn\left( \psi_{\text{OR}}(y_i)\right), \dots) \cdot \prod_{i = 1}^R |\psi_{\text{OR}}(y_i)|.$$
For intuition, let us focus on the final factor in this expression, $\prod_{i = 1}^R |\psi_{\text{OR}}(y_i)|$. Since $\psi_{\text{OR}}$ has $\ell_1$-norm equal to 1, the function $|\psi_{\text{OR}}|$ is a probability distribution, and $\prod_{i = 1}^R |\psi_{\text{OR}}(y_i)|$
is a product distribution over $\left(\bits^{N}\right)^R$. At a high level, our analysis shows that this product distribution is ``exponentially more biased''
toward inputs of low Hamming weight than is $\psi_{\text{OR}}$ itself. 

More specifically, Conditions \eqref{prop1} and \eqref{prop2} together imply that,
if $y=(y_1, \dots, y_R) \in \bits^{N \cdot R}$ is drawn from the product distribution  $\prod_{i = 1}^R |\psi_{\text{OR}}(y_i)|$,
then the probability that $y$ has Hamming weight more than $N = \tilde{O}(R)$ is dominated by the probability that 
roughly $R^{2/3}$ of the $y_i$'s each have Hamming weight close to $R^{1/3}$ (and the remaining $y_i$'s have low
Hamming weight). But then Condition \eqref{prop2} ensures that the probability that this occurs is at most $R^{-\Omega(R^{2/3})}$.

\subsection{Paper Organization}
Section \ref{sec:prelims} covers technical preliminaries.
Stage 1 of the proof of our main hardness amplification theorem, Theorem \ref{thm:main}, is completed in Section \ref{sec:step1}.
In Section \ref{sec:step2}, we execute Stage 2 of the proof of Theorem \ref{thm:main}, and use it to
establish Theorems \ref{thm:acz} and \ref{thm:dnf} from the introduction.
Finally, Section \ref{sec:apps} describes applications of our results to complexity theory and cryptography.

\section{Preliminaries}
\label{sec:prelims}
We begin by formally defining the notion of approximate degree
 of any partial function defined on a subset of $\R^n$.
 Throughout, for any subset $\genericdomain \subseteq \R^n$
 and polynomial $p \colon \genericdomain \to \R$, we
 use $\deg(p)$ to denote the total degree of $p$, and refer
 to this without qualification as the degree of $p$.
 
\begin{definition} \label{def:adeg}
Let $\genericdomain \subseteq \R^n$, and let $f : \genericdomain \to \{-1,1\}$.
The $\eps$-\emph{approximate degree} of $f$, denoted $\adeg_\epsilon(f)$, is the least degree of a real polynomial $p \colon \R^n \to \R$ such that $|p(x)-f(x)| \le \epsilon$ for all $x \in \genericdomain$. We refer to such a $p$ as an $\eps$-\emph{approximating polynomial} for $f$. We use $\adeg(f)$ to denote $\adeg_{1/3}(f)$.
\end{definition}

We highlight two slightly non-standard aspects of Definition \ref{def:adeg}. 
The first is that it considers subsets of $\R^n$ rather than $\bits^n$. This level of generality has been considered in
some prior works \cite{sherstovhalfspaces1, bt16, ambainis}, and we will require it in our proof of Theorem \ref{thm:main} (cf. Section \ref{sec:step1}). Second, our definition of an $\eps$-approximating polynomial $p$ for $f$ above does not place any restriction
on $p(x)$ for $x$ outside of the domain of $f$. This is in contrast to some other works (e.g. \cite{bchtv, comm2, bttoc, comm3}) 
that do require $p(x)$ to be bounded for some inputs $x$ outside of the domain of $f$. Our definition is the most natural and convenient 
for the purposes of our analyses.

Strong LP duality implies the following characterization of approximate degree (see, e.g., \cite{patmat}).

\begin{theorem} \label{thm:duality} \label{thm:dual}
Let $\genericdomain$ be a finite subset of $\R^n$, and let $f : \genericdomain \to \{-1,1\}$.
Then $\adeg_{\eps}(f) \geq d$ if and only if there exists a function 
$\psi \colon \genericdomain \to \R$ satisfying the following properties.
\begin{equation} \label{eq:corr} \sum_{x \in \genericdomain}\psi(x) \cdot f(x) > \eps, \end{equation} 
\begin{equation} \label{eq:unitnorm} \sum_{x \in \genericdomain} |\psi(x)| = 1, and \end{equation} 
\begin{equation} \label{eq:phd} \text{ For every polynomial } p \colon \genericdomain \to \R \text{ of degree less than } d, \sum_{x \in \genericdomain} p(x) \cdot \psi(x) = 0. 
\end{equation}
\end{theorem}

For  functions $\psi_1 \colon \genericdomain \to \R$ and $\psi_2 \colon \genericdomain' \to \R$ 
defined on finite domains $\genericdomain, \genericdomain'$ with $\genericdomain \subseteq \genericdomain'$, 
we define $$\langle \psi_1, \psi_2 \rangle := \sum_{x \in \genericdomain}\psi_1(x) \cdot \psi_2(x),$$ and we refer to this as the correlation of $\psi_1$ with $\psi_2$. (We define $\langle \psi_1, \psi_2 \rangle$ similarly if instead $\genericdomain' \subseteq \genericdomain$.) An equivalent way to define
$\langle \psi_1, \psi_2 \rangle$ is to first extend the domain of $\psi_1$ to $\genericdomain'$ by setting $\psi_1(x) = 0$ for all $x \in \genericdomain'\setminus \genericdomain$,
and then define $$\langle \psi_1, \psi_2 \rangle := \sum_{x \in \genericdomain'}\psi_1(x) \cdot \psi_2(x).$$

We refer to the right hand side of Equation \eqref{eq:unitnorm} as the $\ell_1$-norm of $\psi$, and denote this quantity by $\|\psi\|_1$. If $\psi$ satisfies Equation \eqref{eq:phd}, it is said to have \emph{pure high degree} at least $d$.

\medskip
\noindent \textbf{Additional Notation.} 
For an input $x \in \bits^n$, we use $|x|$ to denote the Hamming weight of $x$, i.e., $|x| := \sum_{i=1}^n (1-x_i)/2$. 
Let $\bits^{N}_{\le k} := \{x \in \bits^N : |x| \le k\}$.
We denote the set $\{1, \dots, N\}$ by $[N]$ and the set $\{0, \dots, N\}$ by $[N]_0$.
Given $t \in \R$, we define $\sgn(t)$ to equal $1$ if $t > 0$ and to equal $-1$ otherwise. The function $\mathbf{1}_N \colon \bits^N \to \bits$
denotes the constant function that always evaluates to 1. We denote by $1^N$ the 
$N$-dimensional vector with all entries equal to $1$.

\medskip
\noindent \textbf{Minsky-Papert Symmetrization.}
The following well-known lemma is due to Minsky and Papert \cite{mp}. 
\begin{lemma} \label{lem:mp}
Let $p \colon \bits^n \to \bits$ be an arbitrary polynomial. Then
there is a univariate polynomial $q \colon \R \to \R$ of degree at most $\deg(p)$
such that 
\[q(t) = \frac{1}{{n \choose t}} \sum_{x \in \bits^n \colon |x|=t} p(x)\]
for all $t \in [n]_0$.
\end{lemma}

\subsection{The Dual Block Method}
This section collects definitions and preliminary results on the dual block method \cite{shizhu, lee, sherstovhalfspaces1}
for constructing dual witnesses for a block composed function $F \circ f$ by combining dual witnesses
for $F$ and $f$ respectively.
\label{sec:ls}
\begin{definition}
Let $\Psi : \bits^M \to \R$ and $\psi : \bits^m \to \R$ be functions that are not identically zero. Let $x=(x_1, \dots, x_M) \in \left(\bits^{m}\right)^M$. Define the \emph{dual block composition} of $\Psi$ and $\psi$, denoted $\Psi \ls \psi : (\bits^m)^M \to \R$, by
\[(\Psi \ls \psi)(x_1, \dots, x_M) = 2^M \cdot \Psi(\dots, \sgn\left(\psi(x_i)\right), \dots) \cdot \prod_{i = 1}^M |\psi(x_i)|.\]
\end{definition}

\begin{proposition} \label{prop:ls}
The dual block composition satisfies the following properties:

\paragraph{\emph{Preservation of $\ell_1$-norm}:} If $\|\Psi\|_1 = 1$ and $\|\psi\|_1 = 1$, then
\begin{equation}\|\Psi \ls \psi\|_1 = 1. \label{eqn:ls-norm}\end{equation}

\paragraph{\emph{Multiplicativity of pure high degree}:} If $\langle \Psi, P \rangle = 0$ for every polynomial $P \colon \bits^M \to \bits$ of degree less than  $D$, and $\langle \psi, p \rangle = 0$ for every polynomial $p \colon \bits^m \to \bits$ of degree less than $d$, then
for every polynomial $q \colon \bits^{m \cdot M} \to \bits$,
\begin{equation}\deg q < D \cdot d \implies \langle \Psi \ls \psi, q \rangle = 0. \label{eqn:ls-phd} \end{equation}

\paragraph{\emph{Associativity}:} For every $\zeta: \bits^{m_\zeta} \to \R$, $\varphi: \bits^{m_\varphi} \to \R$, and $\psi: \bits^{m_\psi} \to \R$, we have
\begin{equation} (\zeta \ls \varphi) \ls \psi = \zeta \ls (\varphi \ls \psi). \label{eqn:ls-assoc} \end{equation}

\end{proposition}

\begin{proof}
Properties~\eqref{eqn:ls-norm} and~\eqref{eqn:ls-phd} appear in~\cite[Proof of Theorem 3.3]{sherstovhalfspaces1}. Proving that
Property~\eqref{eqn:ls-assoc} holds is a straightforward if tedious calculation that we now perform.
 Below, we will write an input $$x = (x_1, \dots, x_{m_{\zeta}}) = ((x_{1, 1}, \dots, x_{1, m_{\varphi}}), \dots, (x_{m_\zeta, 1}, \dots, x_{m_\zeta, m_\varphi})),$$ where each $x_{i, j} \in \bits^{m_{\psi}}$. We expand
\begin{align*}
\!\!\!\!\!\!\!\!\!\!\!\!\!\!\!\!\!\!\!\!\!\!\! ((\zeta &\ls \varphi) \ls \psi)(x_{1, 1}, \dots, x_{m_\zeta, m_\varphi}) = 2^{m_\zeta \cdot m_\varphi} \cdot (\zeta \ls \varphi)(\dots, \sgn\left(\psi(x_{i, j})\right), \dots) \cdot \prod_{i = 1}^{m_\zeta} \prod_{j = 1}^{m_\varphi} |\psi(x_{i, j})| \\
&\!\!\!\!\!\!\!\!\!\!\!\!\!\!\!\!\!= 2^{m_\zeta \cdot m_\varphi} \! \cdot \! \left( \!2^{m_\zeta} \! \cdot \! \zeta (\dots, \sgn \left(\varphi( \dots, \sgn\left(\psi(x_{i, j})\right), \dots)\right), \dots) \!\cdot\! \prod_{i = 1}^{m_\zeta} |\varphi(\dots, \sgn\left(\psi(x_{i, j})\right), \dots)| \! \right) \!\cdot \! \prod_{i = 1}^{m_\zeta} \prod_{j = 1}^{m_\varphi} |\psi(x_{i, j})| \\
&\!\!\!\!\!\!\!\!\!\!\!\!\!\!\!\!\!= 2^{m_\zeta} \cdot \zeta(\dots, \sgn\left(\varphi (\dots, \sgn\left(\psi(x_{i, j})\right), \dots), \dots \right) \cdot \prod_{i = 1}^{m_\zeta} \left( 2^{m_\varphi} \cdot |\varphi(\dots, \sgn\left(\psi(x_{i, j})\right), \dots)| \prod_{j = 1}^{m_{\varphi}} |\psi(x_{i,j})|\right)\\
&\!\!\!\!\!\!\!\!\!\!\!\!\!\!\!\!\!= 2^{m_{\zeta}} \cdot \zeta(\dots, \sgn\left(\left(\varphi \ls \psi\right)(x_i)\right), \dots) \cdot \prod_{i = 1}^{m_\zeta} |(\varphi \ls \psi)(x_{i})| \\
&\!\!\!\!\!\!\!\!\!\!\!\!\!\!\!\!\!= (\zeta \ls (\varphi \ls \psi))(x_1, \dots, x_{m_{\zeta}}).
\end{align*}
\end{proof}

The following proposition identifies conditions under which a dual witness
$\psi$ for the large $(1/3)$-approximate degree of a function
$f$ can be transformed, via dual block composition with a certain function $\Psi \colon \bits^M \to \bits$,
into a dual witness
for the large $(1-2^{-\Omega(M)})$-approximate degree of the block composition $\AND_M \circ f$.

\begin{proposition}[Bun and Thaler~\cite{bt14}] \label{prop:bt-amp}
Let $m, M \in \N$. There exists a function $\Psi : \bits^M \to \R$ with the following properties.
Let $f : \bits^m \to \bits$ be any function. Let $\psi : \bits^m \to \R$ be any function such that $\langle \psi, f \rangle \ge 1/3$, $\|\psi\|_1 = 1$, and $\psi(x) \ge 0$ whenever $f(x) = 1$. Then
\begin{equation}\langle \Psi \ls \psi , \AND_M \circ f \rangle \ge 1 - (2/3)^{M}, \label{eqn:bt-corr}\end{equation}
\begin{equation}\|\Psi \ls \psi\|_1 = 1, \label{eqn:bt-norm}\end{equation}
\begin{equation}\langle \Psi, \mathbf{1}_M \rangle = 0. \label{eqn:bt-phd}\end{equation}
\end{proposition}

The following proposition roughly states that if $\psi$ and $\Psi$ are dual polynomials
that are well-correlated with $f$ and $F$ respectively, then the dual block
composition $\Psi \ls \psi$ is well-correlated with the block composed function
$F \circ f$. There is, however, a potential loss in correlation that is proportional to the number of variables
on which $F$ is defined.
\begin{proposition}[Sherstov~\cite{sherstovhalfspaces1}] \label{prop:sherstov-degree}
Let $f : \bits^m \to \bits$ and $F \colon \bits^M \to \bits$, and let $\eps, \delta > 0$. Let $\psi\colon \bits^m \to \bits$ be a function with $\|\psi\|_1 = 1$ and $\langle \psi, f \rangle \ge 1 - \delta$. Let $\Psi \colon \bits^M \to \bits$ be a function with $\|\Psi\|_1 = 1$ and $\langle \Psi, F \rangle \ge \eps$. Then
\[\langle \Psi \ls \psi, F \circ f \rangle \ge \eps - 4M\delta.\]
\end{proposition}



\section{Connecting Symmetric Properties and Block Composed Functions}
\label{sec:step1}
In this section, we execute Stage 1 of our program for proving our main hardness amplification theorem, Theorem~\ref{thm:main}. Throughout this entire section, fix an arbitrary function $F_R : \bits^R \to \bits$.
(In order to prove Theorem \ref{thm:main}, we will ultimately set $R=10 \cdot n \cdot \log n$, and take $F_R = f \circ \AND_{10 \log n}$ for $f : \bits^n \to \bits$.)

Our analysis relies on several intermediate functions, which we now define and analyze. All of these functions are variants of the function $F_R \circ \OR_N$. 

\subsection{The First Function: Block Composition Under a Promise} We define a promise variant of the function $F_R \circ \OR_N$ as follows.

\begin{definition}
\label{def:promcomp} 
Fix positive numbers $N$ and $R$. Recall that $\bits^{N \cdot R}_{\le N}$ denotes the subset of $\bits^{N \cdot R}$ consisting of 
vectors of Hamming weight at most $N$. 
Define $\promcomp$ to be the partial function obtained from  $F_R \circ \OR_N$ 
by 
restricting its domain to $\bits^{N \cdot R}_{\le N}$. \end{definition}

Our goal is to reduce establishing
Theorem \ref{thm:main} to establishing a lower bound on the approximate degree
of $\promcomp$. Specifically, we prove the following theorem relating the approximate degree of $\promcomp$ to that of a function $g$ which is not much more complex than $F_R$:

\begin{theorem} \label{thm:main-reduction}
 Let $\promcomp \colon \bits_{\le N}^{N \cdot R} \to \bits$ be as in Definition \ref{def:promcomp}. There exists a function\\ $g \colon \bits^{12 \cdot N \cdot \lceil \log(R+1) \rceil} \to \bits$ such that
\begin{equation} \label{sighsighsighh} \adeg_{\eps}(g) \ge \adeg_{\eps}(\promcomp) \cdot \lceil \log(R+1) \rceil. \end{equation}
Moreover:
\begin{align}
\bullet \quad &\text{If $F_R$ is computed by a circuit of depth $k$, then $g$ is computed by a circuit of depth $k+2$.} \label{eqn:reductioneasy}\\
&\nonumber \\
\bullet \quad &\text{If $F_R$ is computed by a monotone circuit of depth $k$, then $g$ is computed by a monotone} \nonumber \\ &\text{circuit of depth $k+2$ with $\AND$ gates at the bottom.}  \label{eqn:reductionaczhard} \\ 
&\nonumber \\
\bullet \quad &\text{If $F_R$ is computed by a monotone DNF of width $w$, then $g$ is computed
by a monotone} \nonumber \\ &\text{DNF of width $O(w \cdot \log R)$.} \label{eqn:reductiondnfhard}
\end{align}
\end{theorem}


\subsection{The Second Function: A Property of Evaluation Tables}
Consider a vector $s=(s_1, \dots, s_N) \in [R]_0^{N}$. Observe that $s$ can be thought of as the 
\emph{evaluation table} of a function $f_s \colon [N] \to [R]_0$ defined
via $f_s(i) = s_i$. The second function $\symcomp$ that we define (cf. Definition \ref{def:symcomp} below)
can be thought of as a \emph{property} of such a function $f_s$. 

In order to define $\symcomp$, 
it is useful to describe such a function $f_s$ as follows.
\begin{definition} \label{def:y}
Fix any $s \in [R]_0^N$. For $(i, j) \in [R]_0 \times [N]$,
define $Y(s)=(\dots, Y_{ij}(s), \dots) \in \left(\bits^{N}\right)^{R+1}$ where
 \[Y_{ij}(s) =  \begin{cases} -1 \text{ if } s_j=i \\
1 \text{ otherwise.} \end{cases} \]
\end{definition}
Observe that any vector $y=(\dots, y_{ij}, \dots) \in \left(\bits^{N}\right)^{R+1}$ equals $Y(s)$ for some $s \in [N]^{R+1}$ if and only if
$y$ satisfies the following condition:
\begin{equation} \label{eq:yprop}
\text{For every } j \in [N], \text{ there exists exactly one value of } i \text{ in } [R]_0 \text{ 
such that } y_{ij} =-1. \end{equation} 
Accordingly, the domain of our second function $\symcomp$ 
is the subset of $\left(\bits^{N}\right)^{R+1}$ satisfying Condition \eqref{eq:yprop}.


\begin{definition}
\label{def:symcomp} 
Let $\domain_{N, R}$ be the subset of $\left(\bits^{N}\right)^{R+1}$ of vectors
satisfying Condition \eqref{eq:yprop}. We refer to any function from $\domain_{N, R}$ to $\bits$
as a \emph{property} of functions $[N] \to [R]_0$.
Define the property $\symcomp \colon \domain_{N, R} \to \bits$ via:
$\symcomp(y_0, y_1, \dots, y_{R}) := F_R(\OR_N(y_1), \dots, \OR_N(y_R)).$
\end{definition}

One may view $\symcomp$ as a property of functions $f_s \colon [N] \to [R]_0$ as follows. The property
$\symcomp$ first obtains a vector of $R$ bits $(b_1, \dots, b_R)$, one for each of the $R$ non-zero range items $1, \dots, R$, 
and then feeds these bits into $F_R$. Here, the bit $b_i$ for range item $i$
is obtained by testing whether $i$ appears in the image of $f_s$ (any occurrences of range item $0$
are effectively ignored by $\symcomp$).

The following lemma establishes that $\symcomp$ satisfies a basic symmetry condition. This holds
regardless of the base function $F_R$ used to define $\symcomp$.

\begin{lemma} \label{lem:symm}
For a permutation $\sigma : [N] \to [N]$ and a vector $y_i \in \bits^{N}$,
let $\sigma(y_i) := (y_{i, \sigma(1)}, \dots, y_{i, \sigma(N)})$. Then
$\symcomp(y_0, \dots, y_{R}) = \symcomp(\sigma(y_0), \dots, \sigma(y_{R})).$
\end{lemma}
\begin{proof}
Immediate from Definition \ref{def:symcomp} and the fact that $\OR_N$
depends only on the Hamming weight of its input.
\end{proof}

Viewing $\symcomp$ as a property of functions $f_s \colon [N] \to [R]_0$,
Lemma \ref{lem:symm} simply states that $\symcomp$
is invariant under permutations of the domain of $f_s$.

\subsection{The Third Function: A Symmetrized Property}
To define our third function $\symsymcomp$,
it is useful to 
consider yet another representation of a function $f_s \colon [N] \to [R]_0$. 
\begin{definition} \label{def:zx}
Given $s \in [R]_0^N$, and its associated function $f_s$, let $Z_i(s)=|f_s^{-1}(i)|$, and define $Z(s) = (Z_0(s), \dots, Z_{R}(s))$.
\end{definition}
That is, each function $Z_i(s)$ counts the number of of inputs $j \in [N]$ such that $f_s(j) = i$.
Observe that a vector $z=(z_0, \dots, z_{R}) \in [N]_0^{R+1}$ equals
$Z(s)$ for some $s \in [R]_0^N$ if and only if
\begin{equation} \label{eq:zprop} z_0 + \dots + z_{R} = N.\end{equation}
Accordingly, the domain upon which our third function $\symsymcomp$ is defined
is the subset of $[N]_0^{R+1}$ satisfying Equation \eqref{eq:zprop}.

\begin{definition}
Let $\tilde{\mathcal{D}}_{N, R}$ be the subset of $[N]_0^{R+1} \subset \R^{R+1}$ consisting of all vectors $z=(z_0, \dots, z_{R})$ satisfying Equation \eqref{eq:zprop}.
Define $\symsymcomp \colon \tilde{\mathcal{D}}_{N, R} \to \bits$ as follows. For any $z \in \tilde{\mathcal{D}}_{N, R}$, let $s$
be an arbitrary vector such that $Z(s)=z$.
Define $\symsymcomp(z) = \symcomp(Y(s))$, where $Y(s)$ is as in Definition \ref{def:y}.
\end{definition}

The function $\symsymcomp$ is well-defined, as Lemma \ref{lem:symm} implies that 
for any pair $s, s' \in [R]_0^N$ such that $Z(s)=Z(s')$,
it holds that $ \symcomp(Y(s)) =  \symcomp(Y(s'))$. 
It is straightforward to see that
the following is an alternative definition of $\symsymcomp$ on its domain $\tilde{\mathcal{D}}_{N, R}$:
\begin{equation} \label{eq:oneuse} \symsymcomp(z_0, \dots, z_{R}) = F_R(\Ind_{> 0}(z_1), \dots, \Ind_{> 0}(z_R)),\end{equation}
where $\Ind_{> 0}(z_i) = 1$ if $z_i = 0$ and is equal to $-1$ otherwise.

\paragraph{Relating the Approximate Degrees of $\symcomp$ and $\promcomp$.} 
The following lemma is implicit in the proof of \cite[Lemma 3.4]{ambainis}. It states that $\symsymcomp$ is no harder
to approximate by low-degree polynomials than is $\symcomp$. 
\begin{lemma}[Ambainis \cite{ambainis}] \label{lem:ambainis}
Let $\symcomp \colon \domain_{N, R} \to \bits$ be any property of functions $f_s \colon [N] \to [R]_0$
that is symmetric with respect to permutations of the domain of $f_s$.
Let $p$ be a polynomial of degree $d$ that $\eps$-approximates $\symcomp$ on its domain $\domain_{N, R}$. 
Then there exists a polynomial $\tilde{p}\colon \R^{R+1} \to \R$ of degree at most $d$ that $\eps$-approximates $\symsymcomp$ 
on its domain $\tilde{\mathcal{D}}_{N, R}$.
\end{lemma}

The following theorem is the technical heart of this section.
\begin{theorem} \label{thm:sym-vs-prom}
Let $\eps > 0$. Then $\adeg_\eps(\symcomp) \ge  \adeg_\eps(\promcomp).$
\end{theorem}

\begin{proof}
Recall that the domain of $\symcomp$ is the subset $\domain_{N, R}$ of $\left(\bits^{N}\right)^{R+1}$ of vectors
satisfying Condition \eqref{eq:yprop}, and the domain of $\promcomp$ is
$\{-1, 1\}_{\le N}^{N \cdot R} = \{x \in (\bits^N)^R : |x| \le N\}$.
Let $p \colon \domain_{N, R} \to \R$ be a polynomial of degree $d$ that $\eps$-approximates $\symcomp$. We will construct a polynomial $q \colon \{-1, 1\}_{\le N}^{N \cdot R} \to \R$ of degree at most $d$ that $\eps$-approximates $\promcomp$. 

By Lemma \ref{lem:ambainis}, there exists a polynomial $\tilde{p} : \R^{R+1} \to \R$ of degree at most $d$ such that
\begin{equation} \label{eq:tp} |\tilde{p}(z_0, \dots, z_{R}) - \symsymcomp(z_0, \dots, z_{R})| \le \eps \qquad \text{ whenever } (z_0, \dots z_{R}) \in \tilde{\domain}_{N,R}.\end{equation}

For each $i = 0, \dots, R$, define a function $T_i : \{-1, 1\}_{\le N}^{N \cdot R}  \to [N]_0$ by $T_i(x) = |\{j \colon x_{ij} = -1\}| = \frac{1}{2}(N - \sum_{j = 1}^N x_{ij})$. Now define the polynomial $q \colon \{-1, 1\}_{\le N}^{N \cdot R}  \to \R$ by
\[q(x) = \tilde{p}\left(N - \sum_{i = 1}^R T_i(x), T_1(x), \dots, T_R(x)\right).\]
Since each of the functions $T_i$ is linear, the polynomial $q$ has degree at most $d$.

We now verify that \begin{equation}
\label{eq:approxeq} |q(x) - \promcomp(x)| \le \eps \text{ for all } x \in \{-1, 1\}_{\le N}^{N \cdot R}. \end{equation} Fix some $x = (x_1, \dots, x_R) \in \{-1, 1\}_{\le N}^{N \cdot R}$. 
Then
\begin{align*}
\symsymcomp\left(N - \sum_{i = 1}^R T_i(x), T_1(x), \dots, T_R(x))\right) &= F_R(\Ind_{> 0}(T_1(x)), \dots, \Ind_{> 0}(T_R(x))) \\
&=F_R(\OR_N(x_1), \dots, \OR_N(x_{R})) \\
&=\promcomp(x).
\end{align*}
Here, the first equality holds by combining Equation \eqref{eq:oneuse}
with the fact that $\left(N - \sum_{ i = 1}^R T_i(x), T_1(x), \dots, T_R(x)\right)$ is a sequence of non-negative numbers summing to $N$ and hence is in the domain $\tilde{\domain}_{N, R}$ of $\symsymcomp$.
The second equality holds by definition of $\Ind_{> 0}$.
The third equality holds by definition of $\promcomp$ and $\{-1, 1\}_{\le N}^{N \cdot R}$.

Property \eqref{eq:approxeq} now follows by definition of $q$ and Property \eqref{eq:tp}.
\end{proof}

\subsection{The Final Function: From a Property to a Circuit}
Recall that our goal in this section is to prove Theorem~\ref{thm:main-reduction} reducing our main hardness
amplification theorem
(Theorem \ref{thm:main}) to a lower bound on the approximate degree
of $\promcomp$. Theorem \ref{thm:main-reduction}
refers to a total function $g$ on $m=O(n \log^4 n)$ bits. But none of the first three functions
defined in this section (i.e., $\promcomp$, $\symcomp$, and $\symsymcomp$) are total functions on bits. 
Hence,
we still need to construct a function $g$ with domain $\bits^{m}$, with circuit depth or DNF width not much higher than that of $F_R$.
(Recall that the function $g$ referred to in Theorem \ref{thm:main} will ultimately 
be obtained in Section \ref{sec:next} by applying the construction here with $R=10 n \log n$, and $F_R = f \circ \AND_{10 \log n}$).

Our function $g$ will interpret its input $u \in \bits^{m}$
as specifying a list $s$ of $N$ numbers from the set $[R]_0$, and will output $\symcomp(Y(s))$.
There are many ways to translate $u$ into the list $s$. It turns out that a relatively
simple translation method suffices to ensure Property \eqref{eqn:reductioneasy} of Theorem \ref{thm:main-reduction}, i.e., that if $F_R$ is computed by a circuit of depth $k$, then $g$ is computed by a circuit of depth $k+2$. We will begin by showing how to construct an auxiliary function $g^*$ that is already enough to satisfy Property~\eqref{eqn:reductioneasy}. Slightly more effort will then be required to modify $g^*$ to construct $g$ establishing Properties \eqref{eqn:reductionaczhard} and \eqref{eqn:reductiondnfhard} of
Theorem \ref{thm:main-reduction}.





\subsubsection{Definition of $g^*$}

\begin{definition} \label{def:shit} \label{def:g}
Fix positive integers $N$, $R$, and $k$ with $k \geq \lceil R+1 \rceil$. 
Let $m = N \cdot k$, and fix any function $\phi \colon \bits^{k} \to [R]_0$. 
We associate an input $u = (u_1, \dots, u_N) \in \left(\bits^k\right)^N$ with
the vector $s_u \in [R]_0^N$ defined as $s_u = (\phi(u_1), \dots, \phi(u_N)).$
Let $Y : [R]_0^N \to \domain_{N,R}$ be as in Definition \ref{def:y}. 
Given any property $G \colon \domain_{N, R} \to \bits$, define
$G_{\phi} \colon \bits^{m} \to \bits$ by $G_{\phi}(u) = G(Y(s_u))$.
\end{definition}


The following lemma is a restatement of  \cite[Theorem 3.2]{sherstov15}.
\begin{lemma}[Sherstov \cite{sherstov15}] \label{lem:phi}
Let $k = 6 \lceil \log (R+1) \rceil$. There exists an (explicitly given) surjection $$\phi: \bits^{k} \to [R]_0$$ such that the following holds.
For every property $G \colon \domain_{N, R} \to \bits$ and $\eps > 0$,
\[\adeg_{\eps}(G_{\phi}) \ge \adeg_{\eps}(G) \cdot \lceil \log (R+1)\rceil.\]
\end{lemma}

The following corollary defines the function $g^*$ and uses Lemma \ref{lem:phi} to show that a lower bound on the approximate degree of
$\promcomp$ implies a lower bound on the approximate degree of $g^*$.
\begin{corollary} \label{cor:reduction}
Fix an integer $N > 0$. Let $\promcomp \colon \bits_{\le N}^{N\cdot R} \to \bits$ be as in Definition \ref{def:promcomp}, $\symcomp \colon
\domain_{N, R} \to \bits$ be as in Definition \ref{def:symcomp}, and $\phi$
be as in Lemma \ref{lem:phi}. Let $m = 6 N \cdot \lceil \log(R+1) \rceil$, and define $g^*  \colon \bits^{m} \to \bits$
to equal $\symcomp_{\phi}$ as per Definition \ref{def:shit}. Then for every $\eps > 0$,
\begin{equation} \label{sighsighsigh} \adeg_{\eps}(g^*) \ge \adeg_{\eps}(\promcomp) \cdot \lceil \log(R+1) \rceil.\end{equation}

Moreover, if $F_R$ is computed by Boolean circuit of depth $k$ and size $\poly(R)$, then $g^*$ is computed by circuit of depth $k+2$ and size $\poly(R, N)$.
\end{corollary}
\begin{proof}
Inequality \eqref{sighsighsigh} follows by combining Lemma \ref{lem:phi} and Theorem \ref{thm:sym-vs-prom}. 

We now turn to showing that if $F_R$ is computed by a circuit of small depth then $g^*$ is as well.
Consider an input $u=(u_1, \dots, u_N) \in \left(\bits^{6 \lceil \log(R+1) \rceil}\right)^N$, and recall that we associate 
$u$ with the vector $(\phi(u_1), \dots, \phi(u_N)) \in [R]_0^N$.
The output $g(u)$ is obtained by applying $F_R$ to a sequence of bits $(b_1, \dots, b_R)$,
where $b_i=-1$ if and only if there exists a $j \in [N]$ such that $\phi(u_j)=i$.

Since $\phi$ is a function on just $6 \lceil \log(R+1) \rceil$ bits, each bit $b_i$ is 
computed by a DNF $\mathcal{C}_i$ of width $6 \lceil \log(R+1) \rceil$, and hence size at most $O(N \cdot R^{12})$.

Hence, if $F_R$ is computed by a Boolean circuit $\mathcal{C}$ of size $S=\poly(R)$ and depth $k$, then by replacing each input $b_i$ to $\mathcal{C}$ with 
the DNF $\mathcal{C}_i$, 
one obtains a circuit $\mathcal{C}^*$ for $g^*$ of size at most $O(S \cdot N \cdot R^{12})=\poly(R, N)$ and depth $k+2$. 
%
\end{proof}

\subsubsection{Definition of $g$}
Even if $F_R$ is a DNF of polylogarithmic width, the function $g^*$ defined in Corollary \ref{cor:reduction}
may not be. However, it is not hard to see that if $F_R$ is a \emph{monotone} DNF of polylogarithmic width $w$,
then $g^*$ is a (non-monotone) DNF of width at most $O(w \cdot \log n)$. Indeed, in this case 
the circuit $\mathcal{C}^*$ for $g^*$ constructed in the proof of Corollary \ref{cor:reduction} is an $\OR-\AND-\OR-\AND$ circuit \emph{with all negations at the input level}.
Each $\AND$ gate in the second level from the top has fan-in at most $w$, and the bottom
 $\AND$ gates each have fan-in at most $w'=6 \lceil \log(R+1) \rceil$.
Any such circuit can be transformed into a (non-monotone) DNF of width at most $w \cdot w' = O(w \cdot \log R)$.

Unfortunately, this observation is still not enough for us to eventually obtain
our desired $n^{1-\delta}$ lower bounds for polylogarithmic width DNFs (cf. Theorem \ref{thm:dnf}).
To obtain such lower bounds, we need to recursively apply our hardness amplification methods, 
and hence we need the harder function $g$ to itself be a \emph{monotone} DNF.  Our definition of $g$
achieves this by applying a simple transformation to $g^*$. This transformation has appeared
in related contexts \cite[Proof of Lemma 3]{kothari}. 

\begin{definition} \label{def:h}
Fix $F_R \colon \bits^R \to \bits$, and let $g^* \colon \bits^{m} \to \bits$ be as in Definition \ref{def:g}. Let $\mathcal{C}^*$
be any circuit computing $g^*$ such that all negations in $\mathcal{C}_g$ appear at the inputs.
Let $g \colon \bits^{2m} \to \bits$ be the monotone function defined as follows.
Associate each of the first $m$ inputs to $g$ with an input to $g^*$, and each of the last $m$ inputs to $g$ with the negation of an input to $g^*$.
Then $g$ is obtained from $g^*$ by replacing each literal of $\mathcal{C}^*$
with the corresponding (unnegated) input to $g$.
\end{definition}

%

We now complete the proof of Theorem~\ref{thm:main-reduction}.

\begin{proof}[Proof of Theorem~\ref{thm:main-reduction}]
We begin by establishing Expression \eqref{sighsighsighh}. 
Let $p \colon \bits^{2m} \to \R$ be a degree $d$ polynomial approximating $g$
to error $\epsilon$. Then one can turn $p$ into a polynomial $q \colon \bits^{m} \to \bits$
of degree at most $d$
approximating $g^*$ to the same error by simply replacing each input to $p$ with the corresponding input (or its negation) to
$g^*$. It follows that $ \adeg_{\eps}(g) \ge \adeg_{\eps}(g^*)$. The inequality in Expression \eqref{sighsighsighh} follows from Corollary \ref{cor:reduction}.



Property~\eqref{eqn:reductioneasy} is immediate from Corollary~\ref{cor:reduction}, since the construction of Definition~\ref{def:h} does not change the circuit depth of $\mathcal{C}^*$.

The discussion preceding the statement of Definition \ref{def:h} revealed that, if $F_R$ is a monotone DNF of polylogarithmic width $w$,
then $g^*$ is a (non-monotone) DNF of width at most $O(w \cdot \log R)$. It is then immediate from the
definition of $g$ that $g$ is a monotone DNF of width $O(w \cdot \log R)$. This yields Property~\eqref{eqn:reductiondnfhard}.

By similar reasoning, if $F_R$ is computed by a monotone circuit of depth $k$, 
then $g^*$ is computed by a circuit of depth $k+2$ 
with $\AND$
gates at the bottom, and all negations at the inputs. 
It is then immediate from the
definition of $g$ that $g$ is computed by a monotone circuit of depth $k+2$ with 
$\AND$
gates at the bottom. This establishes Property~\eqref{eqn:reductionaczhard}, completing the proof.
\end{proof}

\section{Analyzing Block Composed Functions On Low Hamming Weight Inputs}
\label{sec:step2} \label{sec:next}
To complete the proof of Theorem \ref{thm:main}, we combine the following 
theorem with Theorem~\ref{thm:main-reduction}.

\begin{theorem}   \label{thm:main-amp}
Let $f_n : \bits^n \to \bits$ be any function. Let $N = c \cdot n \log^3 n$ for a sufficiently large constant $c>0$. Let $\promcomp\colon \{-1, 1\}^{10 \cdot N \cdot n\cdot \log n}_{\le N} \to \bits$ equal $f_n \circ \AND_{10\log n} \circ \OR_N$ restricted to inputs in $\{-1, 1\}^{10 \cdot N \cdot n\cdot \log n}_{\le N} = \{x \in \bits^{10 \cdot N \cdot n\cdot \log n} : |x| \le N\}$ (cf. Definition \ref{def:promcomp}). Then $\adeg(\promcomp) \ge n^{1/3} \cdot \adeg(f_n)^{2/3}.$
\end{theorem}

The primary goal of this section is to prove Theorem~\ref{thm:main-amp}. Before embarking on this proof,
we use it to complete the proofs of Theorems \ref{thm:acz}-\ref{thm:main} from
Section \ref{sec:theorems}.

\medskip
\begin{proof}[Proof of Theorem \ref{thm:main}
assuming Theorem \ref{thm:main-amp}.]
We begin by establishing Property \eqref{propeasy} in the conclusion of
Theorem \ref{thm:main}.
Let $R=10 \cdot n \cdot  \log n$ and $F_R := f_n \circ \AND_{10 \cdot \log n}$.
Applying Corollary \ref{cor:reduction} to $F_R$
yields a function $g$  on $O(N \log R) = O(n \log^4 n)$ variables satisfying 
$$\adeg_{\eps}(g) \ge \adeg_{\eps}(\promcomp) \cdot \lceil \log(R+1) \rceil \geq \Omega\left( n^{1/3} \cdot \adeg(f_n)^{2/3} \cdot \log n\right),$$
where the final inequality holds by Theorem \ref{thm:main-amp}.
Suppose $f_n$ is computed by polynomial size Boolean circuit $\mathcal{C}$ of depth $k$. 
Then $F_R$ is computed by polynomial-size Boolean circuit of depth $k+1$, 
and Property~\eqref{eqn:reductioneasy} of Theorem~\ref{thm:main-reduction} guarantees that $g$ is computed by polynomial size Boolean circuit of depth $k+3$.
Hence, $g$ satisfies Property \eqref{propeasy} as desired.


To establish Property \eqref{propaczhard}, suppose that $f_n$ is computed by a monotone
circuit of depth $k$ with 
$\AND$ gates at the bottom. 
Then $F_R$ is computed by such a circuit as well.
Property~\eqref{eqn:reductionaczhard} of Theorem~\ref{thm:main-reduction} then implies that $g$ is computed by 
a circuit of depth $k+2$ with 
$\AND$ gates at the bottom. 

To establish Property \eqref{propdnfhard}, observe that if $f_n$ is computed
by a monotone DNF of width $w$, then $F_R$ is computed by a monotone DNF of width $O(w \cdot \log n)$,
and Property~\eqref{eqn:reductiondnfhard} of Theorem~\ref{thm:main-reduction} implies that $g$ is computed by a monotone DNF of width $O(w \cdot \log^2 n)$.

\end{proof}

\medskip
\begin{proof}[Proof of Theorems \ref{thm:acz} and \ref{thm:dnf}
assuming Theorem \ref{thm:main}]
One can almost obtain Theorems \ref{thm:acz} and \ref{thm:dnf}
by recursively applying Theorem \ref{thm:main}, starting in the base case with the function $\OR_n$. 
However, to obtain stronger degree lower bounds for a given circuit depth
or DNF width, we instead
use the following well-known result of Aaronson and Shi \cite{aaronsonshi} regarding
the approximate degree of (the negation of) the well-known
Element Distinctness function.

\begin{lemma}[Sherstov \cite{sherstov15}, refining Aaronson and Shi \cite{aaronsonshi}] \label{thm:ed}
There is a function $\overline{\ED} \colon \bits^n \to \bits$ such that
$\adeg(\overline{\ED}) = \Omega(n^{2/3} \log^{1/3} n)$. Moreover, $\overline{\ED}$ 
is computed by a monotone DNF of polynomial size and width $O(\log n)$. 
\end{lemma}
\begin{proof}
If the word monotone were omitted from the conclusion, this would be a restatement of \cite[Theorem 3.3]{sherstov15}.
Using the same technique as in Definition \ref{def:h}, the non-monotone DNF constructed in \cite[Theorem 3.3]{sherstov15} can be transformed
into a monotone DNF satisfying the same properties.
\end{proof}

Lemma \ref{thm:ed} immediately implies Theorems \ref{thm:acz} and \ref{thm:dnf}
in the case $k=1$. 
Assume by way of induction that Theorem \ref{thm:acz} holds for an integer $k\geq 1$. That is, there exists a function $f^{(k)}$ on $n \cdot \log^{4k-4}(n)$ variables, computed by monotone circuit of depth $2k$  with 
$\AND$ gates at the bottom, with approximate degree $\Omega(n^{1-2^{k-1}/3^k} \cdot \log^{3 - 2^{k+2}/3^{k}}(n))$. 
By applying Theorem \ref{thm:main} to $f^{(k)}$, one obtains (by Property \eqref{propaczhard}) 
a function $f^{(k+1)}$ on $n \cdot \log^{4k}(n)$ variables, computed by a monotone circuit  of depth $2k+2$ with 
$\AND$ gates at the bottom, with approximate degree $\Omega(n^{1-2^{k}/3^{k+1}} \cdot \log^{3 - 2^{k+3}/3^{k+1}}(n))$.
The function $f^{(k+1)}$ satisfies the conclusion of Theorem \ref{thm:acz}, completing the inductive proof of Theorem \ref{thm:acz}.

Similarly, assume by way of induction that Theorem \ref{thm:dnf}
holds for an integer $k\geq 1$, for a DNF $f^{(k)}$. 
By applying Theorem \ref{thm:main} to $f^{(k)}$, one obtains (by Property \eqref{propdnfhard}) 
a function $f^{(k+1)}$
satisfying the conclusion of Theorem \ref{thm:dnf} for integer $k+1$.
\end{proof}

\subsection{Organization of the Proof of Theorem \ref{thm:main-amp}}

Our proof of Theorem~\ref{thm:main-amp} entails using a dual witness for the approximate degree of $f_n$ to construct a dual witness for the higher approximate degree of $\promcomp$. For expository purposes, we think about the construction of a dual witness for $\promcomp$ as consisting of four steps.

\paragraph{Step 1.} Let $d=\adeg(f_n)$. We begin by constructing a dual witness $\varphi$ for the $\Omega\left(\sqrt{k}\right)$-approximate degree of the $\OR_N$ function when restricted to inputs of Hamming weight at most $k = (n/d)^{2/3}$. This construction closely mirrors previous constructions of \v{S}palek~\cite{spalek} and Bun and Thaler~\cite{bt14}. However, we need $\varphi$ to satisfy an additional metric condition that is not guaranteed by these prior constructions. Specifically, we require that the total $\ell_1$ weight that $\varphi$ places on the $t$'th layer of the Hamming cube should be upper bounded by $O(1 /(t+1)^2)$.

\paragraph{Step 2.} We apply the error amplification construction of Proposition~\ref{prop:bt-amp} to transform $\varphi$ into a new dual polynomial $\psi$ that witnesses the fact that the $(1-\delta)$-approximate degree of the function $\AND_{10\log n} \circ \OR_N$ remains $\Omega(\sqrt{k})$, even with error parameter $\delta \le 1/N^2$.

\paragraph{Step 3.} We appeal to the degree amplification construction of Proposition~\ref{prop:sherstov-degree} to combine $\psi$ from Step 2 with a dual witness $\Psi$ for the high approximate degree of $f_n$. This yields a dual witness $\zeta$ showing that the approximate degree of the composed function $f_n \circ \AND_{10\log n} \circ \OR_N$ is $\Omega(d \cdot \sqrt{k}) = \Omega(n^{1/3}\cdot d^{2/3})$.

\paragraph{Step 4.} Using a construction of Razborov and Sherstov~\cite{sherstovrazborov}, we zero out the mass that $\zeta$ places on inputs of  Hamming weight larger than $N$, while maintaining its pure high degree and correlation with $\promcomp$. This yields the final desired dual witness $\hat{\zeta}$ for $\promcomp$.

\subsection{Step 1: A Dual Witness for $\OR_N$}

\begin{proposition} \label{prop:or-dual}
Let $k, N \in \N$ with $k \le N$. 
Then there exist a constant $c_1 \in (0, 1)$ and a function $\psi : \bits^N_{\le k} \to \bits$ such that:
\begin{equation}\langle \psi, \OR_N \rangle \ge 1/3  \label{eqn:or-corr}\end{equation}
\begin{equation}\|\psi\|_1 = 1   \label{eqn:or-norm}\end{equation}
\begin{equation} \text{For any polynomial } p \colon \bits^N \to \R\text{, } \deg p < c_1\sqrt{k} \implies \langle \psi, p \rangle = 0 \label{eqn:or-phd}\end{equation}
\begin{equation}\psi(1^N) > 0 \label{eqn:or-onesided}\end{equation}
\begin{equation}\sum_{|x| = t} |\psi(x)| \le 5/(t+1)^2 \qquad \forall t = 0, 1, \dots, k\label{eqn:or-decay}\end{equation}
\end{proposition}

For intuition, we mention that Properties~\eqref{eqn:or-corr}-\eqref{eqn:or-onesided} amount to a dual formulation of the fact that the ``one-sided'' approximate degree  of $\OR_N$ is $\Omega(\sqrt{k})$, even under the promise that the input has Hamming weight at most $k$.\footnote{One-sided approximate degree is a variant of approximate degree defined in, e.g., \cite{bt14}. We will not need the primal formulation of approximate degree in this work, and therefore omit a formal definition of this notion.} Property~\eqref{eqn:or-decay} is an additional metric condition that we require later in the proof.

The key to proving Proposition~\ref{prop:or-dual} is the following explicit construction of a univariate function from first principles. The construction closely follows previous work of \v{S}palek~\cite{spalek} and Bun and Thaler~\cite{bt14}, and appears in Appendix~\ref{app:or-dual}.

\begin{lemma} \label{lem:or-sym-dual} \label{lemma:omega}
Let $k \in \N$. There exists a constant $c_1 \in (0, 1) $ and a function $\omega : \{0, 1, \dots, k\} \to \R$ such that
\begin{equation}\omega(0) - \sum_{t = 1}^k \omega(t) \ge 1/3  \label{eqn:or-sym-corr}\end{equation}
\begin{equation} \sum_{t=0}^k |\omega(t)| = 1   \label{eqn:or-sym-norm}\end{equation}
\begin{equation}  \text{For all univariate polynomials } q \colon \R \to \R\text{, } \deg q < c_1\sqrt{k} \implies \sum_{t=0}^k \omega(t) \cdot q(t) = 0 \label{eqn:or-sym-phd}\end{equation}
\begin{equation} \omega(0) > 0 \label{eqn:or-sym-onesided}\end{equation}
\begin{equation} \omega(t) \le 5 / (t+1)^2  \qquad \forall t = 0, 1, \dots, k\label{eqn:or-sym-decay}\end{equation}
\end{lemma}

\begin{proof}[Proof of Proposition~\ref{prop:or-dual}]
Let $\omega$ be the function guaranteed by Lemma~\ref{lem:or-sym-dual}. Consider the function $\psi : \bits^N_{\le k} \to \bits$ defined by
\[\psi(x) = \frac{1}{\binom{N}{|x|}} \cdot \omega(|x|).\]
That $\psi$ satisfies Conditions \eqref{eqn:or-corr},  \eqref{eqn:or-norm},  \eqref{eqn:or-onesided}, and  \eqref{eqn:or-decay} is immediate from the definition of $\psi$ and Properties \eqref{eqn:or-sym-corr},  \eqref{eqn:or-sym-norm},  \eqref{eqn:or-sym-onesided}, and  \eqref{eqn:or-sym-decay} of $\omega$. Property  \eqref{eqn:or-phd} is a consequence of Minsky-Papert symmetrization. Specifically, for any polynomial $p \colon \bits^N \to \R$, Lemma \ref{lem:mp}
implies that there is a univariate polynomial $q$
of degree at most $\deg(p)$ such that for all $t \in [N]_0$, we have $q(t) = {N \choose t}^{-1} \sum_{x \in \bits^N \colon |x|=t} p(x)$. 
Hence, $\sum_{x \in \bits^N} \psi(x) \cdot p(x) = \sum_{t=0}^N \omega(t) \cdot q(t) = 0$,
where the final equality holds by Property \eqref{eqn:or-sym-phd}.
\end{proof}

\subsection{Steps 2 and 3: A Preliminary Dual Witness for $G = f_n \circ \AND_{10 \log n} \circ \OR_N$} \label{sec:prelim-dual}
Recall that our ultimate goal in this section is to construct a dual witness for the veracity of Theorem \ref{thm:main-amp}. Here, we begin by defining a preliminary dual witness 
$\zeta$. While $\zeta$ itself is insufficient to witness the veracity of Theorem \ref{thm:main-amp},
we will ultimately ``post-process'' $\zeta$ into the desired dual witness $\hat{\zeta}$.
\label{sec:parametersetting} 
We start by fixing choices of several key parameters:
\begin{itemize}
\item $d=\adeg_{2/3}(f_n)$.
\item $k = \lfloor(n/d)^{1/3}\rfloor^2$
\item $D = c_1 \sqrt{k} \cdot d =O(n^{1/3} \cdot d^{2/3})$, where $c_1$ is the constant from Lemma~\ref{lem:or-sym-dual}
\item $R = 10 n \log n$
\item $N = \lceil c_2 R \log^2 R\rceil$, where $c_2$ is a universal constant to be determined later (cf. Proposition \ref{prop:32})
\item $m = R \cdot N$
\end{itemize}
To state our construction of a preliminary dual witness $\zeta$, we begin with the following objects:

\begin{itemize}
\item A dual witness $\varphi : \bits^n \to \R$ for the fact that $\adeg_{2/3}(f_n) \ge d$. By Theorem \ref{thm:dual},
$\varphi$ satisfies the following conditions.
\begin{equation}\langle \varphi, f_n\rangle \ge 2/3    \label{eqn:f-corr} \end{equation}
\begin{equation}\|\varphi\|_1 = 1 \label{eqn:f-norm} \end{equation}
\begin{equation} \text{For any polynomial } p \colon \bits^n \to \R\text{, }\deg p < d \implies \langle \varphi, p \rangle = 0 \label{eqn:f-phd} \end{equation}
\item The function $\Psi : \bits^{10 \log n} \to \R$ whose existence is guaranteed by Proposition~\ref{prop:bt-amp}.
\item The dual witness $\psi : \bits^N \to \R$ for $\OR_N$ guaranteed by Proposition~\ref{prop:or-dual}, using the choice of the parameter $k$ above.
\end{itemize}

We apply dual block composition sequentially to the three dual witnesses to obtain a function $\zeta = \varphi \ls \Psi \ls \psi$. This function is well-defined because dual block composition is associative (Proposition~\ref{prop:ls}).

\begin{proposition}   \label{prop:prelim-basic}
The dual witness $\zeta = \varphi \ls \Psi \ls \psi$ satisfies the following properties:
\begin{equation}\langle \zeta, G \rangle \ge 1/2  \label{eqn:prelim-corr}\end{equation}
\begin{equation}\|\zeta\|_1 = 1   \label{eqn:prelim-norm}\end{equation}
\begin{equation}\text{For all polynomials } p \colon ((\bits^N)^{10 \log n})^n \to \R\text{, }\deg p < D \implies \langle \zeta, p \rangle = 0. \label{eqn:prelim-phd}\end{equation}
\end{proposition}

\begin{proof}
It is easiest to reason about these properties by regarding $\zeta$ as $\varphi \ls (\Psi \ls \psi)$. To this end, let $\xi : (\bits^N)^{10 \log n} \to \R$ denote $\Psi \ls \psi$. Then $\xi$ satisfies the following properties:
\begin{equation}\langle \xi, \AND_{10 \log n} \circ \OR_N \rangle \ge 1 - \frac{1}{24n}  \label{eqn:amp-corr}\end{equation}
\begin{equation}\|\xi\|_1 = 1  \label{eqn:amp-norm}\end{equation}
\begin{equation}\deg p < c_1\sqrt{k} \implies \langle \xi, p \rangle = 0. \label{eqn:amp-phd}\end{equation}

Property~\eqref{eqn:amp-corr} follows from Expression~\eqref{eqn:bt-corr} of Proposition~\ref{prop:bt-amp}, together with Properties~\eqref{eqn:or-corr}, \eqref{eqn:or-norm}, and~\eqref{eqn:or-onesided} of the dual witness $\psi$ for $\OR_N$. Property~\eqref{eqn:amp-norm} follows from Property~\eqref{eqn:ls-norm} of dual block composition (cf. Proposition~\ref{prop:ls}), and the fact that both $\Psi$ and $\psi$ have unit $\ell_1$-norm (Equations~\eqref{eqn:bt-norm} and~\eqref{eqn:or-norm}). Finally, Property~\eqref{eqn:amp-phd} is a consequence of Property~\eqref{eqn:ls-phd} of dual block composition  (cf. Proposition~\ref{prop:ls}), together with Properties~\eqref{eqn:bt-phd} and~\eqref{eqn:or-phd}, which state that $\Psi$ and $\psi$ have pure high degree at least $1$ and $c_1 \sqrt{k}$, respectively.

We now verify Properties~\eqref{eqn:prelim-corr}-\eqref{eqn:prelim-phd} of $\zeta = \varphi \ls \xi$. Property~\eqref{eqn:prelim-corr} follows from~Proposition~\ref{prop:sherstov-degree}, together with Properties~\eqref{eqn:f-corr} and~\eqref{eqn:f-norm} of $\varphi$ and Properties~\eqref{eqn:amp-corr} and~\eqref{eqn:amp-norm} of $\xi$. Property~\eqref{eqn:prelim-norm} follows from Property~\eqref{eqn:ls-norm} of dual block composition (cf. Proposition~\ref{prop:ls}), and the fact that both $\varphi$ and $\xi$ have unit $\ell_1$-norm (Equations~\eqref{eqn:f-norm} and~\eqref{eqn:amp-norm}). Finally, Property~\eqref{eqn:prelim-phd} follows from Property~\eqref{eqn:ls-phd} (cf. Proposition~\ref{prop:ls}) of dual block composition, together with Properties~\eqref{eqn:f-phd} and~\eqref{eqn:amp-phd} of the pure high degrees of $\varphi$ and $\xi$, respectively.
\end{proof}

\subsection{Step 4: Constructing the Final Dual Witness} \label{sec:final-dual}

For a fixed number $N \in \mathbb{N}$, let 
$X = \bits^{N\cdot 10\log n \cdot n}_{\le N} = \{x \in ((\bits^N)^{10\log n})^n : |x| \le N\}.$
Recall that this set $X$ is the same one that appears in Definition \ref{def:promcomp} when applied to the function $F_R := f_n \circ \AND_{10 \log n}$ on $R=10 n \log n$ variables.

\begin{proposition} \label{prop:prelim-mass} \label{prop:32}
Let $\zeta : ((\bits^N)^{10\log n})^n \to \R$ be as constructed in Proposition \ref{prop:prelim-basic}. Then there exists a constant $c_2 > 0 $ such that, for $N = \lceil c_2 R \log^2 R \rceil$ and sufficiently large $n$,
\begin{equation}
\sum_{x \notin X} |\zeta(x)| \le (2NR)^{-2 R / k} \le (2NR)^{-2D}. \label{eqn:prelim-mass}
\end{equation}
\end{proposition}

\begin{proof}
For the proof of Proposition~\ref{prop:prelim-mass}, it is now useful to regard the dual witness $\zeta$ as the iterated dual block composition $(\varphi \ls \Psi) \ls \psi$. In this proof, let us denote $\Phi := \varphi \ls \Psi$. Then $\Phi : \bits^R \to \R$ where $R = 10 n \log n$.

Write $\psi$ as a difference of non-negative functions $\psi_{+1} - \psi_{-1}$. Since $\psi$ has strictly positive pure high degree, 
it is in particular orthogonal to the constant function $\mathbf{1}_N$, and hence $\|\psi_{+1}\|_1 = \|\psi_{-1}\|_1 = 1/2$. Recalling that $\psi(x) = \omega(|x|) / \binom{N}{|x|}$ where $\omega \colon [k]_{0} \to \R$ is given in Lemma \ref{lemma:omega}, we may analogously write $\omega = \omega_{+1} - \omega_{-1}$  where $\omega_{+1}$ and $\omega_{-1}$ are non-negative
functions satisfying
\begin{equation} \label{eq:sweeteq} \sum_{t=0}^k \omega_{+1}(t) = \sum_{t=0}^k \omega_{-1}(t) = 1/2. \end{equation}

By the definition of dual block composition, we have
\[\zeta(x_1, \dots, x_R) = 2^R \cdot \Phi(\dots, \sgn\left(\psi(x_i)\right), \dots) \cdot \prod_{i = 1}^R |\psi(x_i)|.\]
Consequently,
\begin{align}
\notag \sum_{x \notin X} |\zeta(x)| &= 2^R \sum_{z \in \bits^R} |\Phi(z)| \left( \sum_{\substack{(x_1, \dots, x_R) \notin X \text{ s.t.} \\ \notag \sgn\left( \psi(x_1)\right) = z_1, \dots, \sgn\left(\psi(x_R)\right) = z_R}} \prod_{i = 1}^R |\psi(x_i)| \right) \\
\notag &= 2^R \sum_{z \in \bits^R} |\Phi(z)| \left( \sum_{(x_1, \dots, x_R) \notin X} \prod_{i = 1}^R \psi_{z_i}(x_i) \right)\\
&= 2^R \sum_{z \in \bits^R} |\Phi(z)| \left( \sum_{(x_1, \dots, x_R) \notin X} \prod_{i = 1}^R \frac{\omega_{z_i}(|x_i|)}{{N \choose |x_i|}} \right). \label{finallineman}
\end{align}

Observe that for any $(t_1, \dots, t_R) \in [k]_0^R$, the number of inputs $(x_1, \dots, x_R) \in \left(\bits^N\right)^R$
such that $|x_i| = t_i$ for all $i \in [R]$ is exactly $\prod_{i=1}^{R} {N \choose t_i}$. 
Hence, defining $$P = \{(t_1, \dots, t_R) \in [k]_0^R : t_1 + \dots + t_R > N\},$$ we
may rewrite Expression
\eqref{finallineman} as
\[2^R \sum_{z \in \bits^R} |\Phi(z)| \left( \sum_{(t_1, \dots, t_R) \in P} \prod_{i = 1}^R \omega_{z_i}(t_i) \right).\]
To control this quantity, we appeal to the following combinatorial lemma, whose proof we defer to Section~\ref{sec:combinatorial}.

\begin{restatable}{lemma}{combinatorial} \label{lem:combinatorial}
Let $k, R \in \N$ with $k \le N$. There is a constant $\alpha > 0$ such that the following holds. Let $N = \lceil \alpha R \log^2 R \rceil$. Let $\eta_i :[k]_0 \to \R$, for $i = 1, \dots R$, be a sequence of non-negative functions where for every $i$,
\begin{equation}\sum_{r = 0}^k \eta_i(r) \le 1 / 2     \label{eqn:omega-bound} \end{equation}
\begin{equation} \eta_i(r) \le 5/(r+1)^2  \qquad \forall r = 0, 1, \dots, k.     \label{eqn:omega-decay} \end{equation}
Let $P = \{\vec{t} = (t_1, \dots, t_R) \in [k]_0^R : t_1 + \dots + t_R > N\}$. Then
\[\sum_{\vec{t} \in P} \prod_{i = 1}^R  \eta_i(t_i) \le 2^{-R} \cdot (2NR)^{-2 R/k}.\]
\end{restatable}

Observe that the functions $\omega_{z_i}$ satisfy Condition \eqref{eqn:omega-bound} (cf. Equation \eqref{eq:sweeteq}) and Condition \eqref{eqn:omega-decay} (cf. Property~\eqref{eqn:or-sym-decay}). We complete the proof of Proposition~\ref{prop:prelim-mass} by letting $c_2$ equal the constant $\alpha$ appearing in the statement of Lemma~\ref{lem:combinatorial}, and bounding
\begin{align*}
2^R \sum_{z \in \bits^R} |\Phi(z)| \left( \sum_{\vec{t} \in P} \prod_{i = 1}^R \omega_{z_i}(t_i) \right) &\le 2^R \sum_{z \in \bits^R} |\Phi(z)| \cdot \left(2^{-R} \cdot (2NR)^{-2R /k}\right)\\
&=(2NR)^{-2 R / k} \leq (2NR)^{-2D}.
\end{align*}
Here, the equality appeals to the fact that $\|\Phi\|_1 = 1$ (by Property \eqref{eqn:ls-norm} of Proposition \ref{prop:ls}), and the last inequality holds for sufficiently large $n$ by virtue of the fact that $R/k = \Theta(n^{1/3} d^{2/3} \log n)$,
while $D=O(n^{1/3} d^{2/3})$ for the values of $R$ and $D$ specified at the start of Section \ref{sec:parametersetting}.
\end{proof}

We are now in a position to construct our final dual witness for the high approximate degree of $\promcomp$. This dual witness $\hat{\zeta}$ is obtained by modifying $\zeta$ to zero out all of the mass it places on inputs of total Hamming weight larger than $N$. This zeroing process is done in a careful way so as not to decrease the pure high degree of $\zeta$, nor to significantly affect its correlation with $\promcomp$. The technical tool that enables this process is a construction of Razborov and Sherstov~\cite{sherstovrazborov}.

\begin{lemma}[cf. {\cite[Proof of Lemma 3.2]{sherstovrazborov}}]\label{lem:rs}
Let $D, m \in \N$ with $0 \le D \le m - 1$. Then for every $y \in \bits^m$ with $|y| > D$, there exists a function $\phi_y: \bits^m \to \mathbb{R}$ such that
\begin{equation}\phi_y(y) = 1	\label{eqn:rs-point}\end{equation}
\begin{equation}|x| > D, x \ne y \implies \phi_y(x) = 0	\label{eqn:rs-hamming}\end{equation}
\begin{equation}\deg p < D \implies \langle \phi_y, p \rangle = 0 \label{eqn:rs-phd} \end{equation}
\begin{equation}\sum_{|x| \le D} |\phi_y(x)| \le 2^D {|y| \choose D}.   \label{eqn:rs-norm} \end{equation}
\end{lemma}

\begin{proposition} \label{prop:correction}
There exists a function $\nu : ((\bits^N)^{10 \log n})^n \to \R$ such that
\begin{equation}    \text{For all polynomials } p \colon ((\bits^N)^{10 \log n})^n  \to \R\text{, } \deg p < D \implies \langle \nu, p \rangle = 0     \label{eqn:correction-phd}  \end{equation}
\begin{equation}    \|\nu\|_1 \le 1/10     \label{eqn:correction-norm}  \end{equation}
\begin{equation}    |x| > N \implies \nu(x) = \zeta(x),     \label{eqn:correction-correct}  \end{equation}
where $\zeta$ is as in Proposition \ref{prop:prelim-basic}.
\end{proposition}

\begin{proof}
Define
\[\nu(x) = \sum_{y : |y| > N} \zeta(y) \phi_y(x),\]
where $\phi_y$ is as in Lemma~\ref{lem:rs} with $m$ and $D$ set as at the beginning of Section~\ref{sec:prelim-dual}. Property~\eqref{eqn:correction-phd} follows immediately from Property~\eqref{eqn:rs-phd} and linearity. By Proposition~\ref{prop:prelim-mass} and Property~\eqref{eqn:rs-norm}, we have
\begin{align*}
\|\nu\|_1 &\le \sum_{y : |y| > N} |\zeta(y)| \cdot 2^D \cdot {|y|\choose D} \\
&\le (2NR)^{-2D} \cdot 2^D \cdot m^D \\
&\le (2m)^{-2D} \cdot (2m)^D \\
&\le 1/10,
\end{align*}
establishing Property~\eqref{eqn:correction-norm}. Finally, Property~\eqref{eqn:correction-correct} follows from~\eqref{eqn:rs-point} and~\eqref{eqn:rs-hamming}, together with the fact that $D < N$.
\end{proof}

Combining Proposition~\ref{prop:correction} with Proposition~\ref{prop:prelim-basic} allows us to complete the proof of Theorem~\ref{thm:main-amp}, which was the goal of this section.

\begin{proof}[Proof of Theorem~\ref{thm:main-amp}]
Let $\zeta = \varphi \ls \Psi \ls \psi$ be as defined in Section~\ref{sec:prelim-dual}, and let $\nu$ be the correction object constructed in Proposition~\ref{prop:correction}. Observe that $\|\zeta - \nu\|_1 > 0$, as
$\|\zeta\|_1 = 1$ (cf. Equality \eqref{eqn:prelim-norm}) and $\|\nu\|_1 \leq 1/10$ (cf. Inequality \eqref{eqn:correction-norm}).
Define the function
\[\hat{\zeta}(x) = \frac{\zeta(x) - \nu(x)}{\|\zeta - \nu\|_1}.\]
Since $\nu(x) = \zeta(x)$ whenever $|x| > N$ (cf. Equation~\eqref{eqn:correction-correct}), the function $\hat{\zeta}$ is supported on the set $X$. By Theorem \ref{thm:dual}, to show that it is a dual witness for the high approximate degree of $\promcomp$, it suffices to show that $\hat{\zeta}$ satisfies the following three properties:
\begin{equation} \langle \hat{\zeta}, \promcomp \rangle \ge 1/3 \label{eqn:final-corr} \end{equation}
\begin{equation} \| \hat{\zeta} \|_1 = 1 \label{eqn:final-norm} \end{equation}
\begin{equation}  \text{For all polynomials } p \colon ((\bits^N)^{10 \log n})^n  \to \R\text{, } \deg p < D \implies \langle \hat{\zeta}, p\rangle = 0. \label{eqn:final-phd} \end{equation}
We establish~\eqref{eqn:final-corr} by computing
\begin{align*}
\langle \hat{\zeta}, \promcomp \rangle &= \frac{1}{ \|\zeta - \nu\|_1} \langle \zeta - \nu, \promcomp \rangle \\
&= \frac{1}{ \|\zeta - \nu\|_1}  \langle \zeta - \nu, G \rangle & \text{since $\zeta = \nu$ outside $X$} \\
&= \frac{1}{ \|\zeta - \nu\|_1}  \left( \langle \zeta, G \rangle - \langle \nu, G \rangle\right) & \\
&\ge \frac{1}{ \|\zeta - \nu\|_1} \left( \langle \zeta, G \rangle - \|\nu\|_1 \right) &\\ 
&\ge \frac{1}{ \|\zeta - \nu\|_1} \left( 1/2 - 1/10 \right)  & \text{by~\eqref{eqn:prelim-corr} and~\eqref{eqn:correction-norm}}\\ 
&\ge \frac{1}{ \|\zeta\|_1 + \|\nu\|_1} \left( 1/2 - 1/10 \right)  &\\ 
&\ge \frac{1}{1 + 1/10} \left( 1/2 - 1/10 \right) & \text{by~\eqref{eqn:prelim-norm} and~\eqref{eqn:correction-norm}} \\
&\ge \frac{1}{3}.
\end{align*}
Equation~\eqref{eqn:final-norm} is immediate from the definition of $\hat{\zeta}$. Finally,~\eqref{eqn:final-phd} follows from~\eqref{eqn:prelim-phd}, \eqref{eqn:correction-phd}, and linearity.

\end{proof}

\subsection{Proof of Lemma~\ref{lem:combinatorial}}   \label{sec:combinatorial}

All that remains to complete the proof of Theorem~\ref{thm:main-amp} is to establish the deferred combinatorial lemma from Section~\ref{sec:final-dual}. We begin by stating a two simple lemmas.

\begin{lemma} \label{lem:binomial}
Let $k, n \in \N$ with $k \le n$. Then 
$\binom{n}{k} \le \left(\frac{en}{k}\right)^k$.
\end{lemma}

\begin{lemma} \label{lem:inv-sq-tail}
Let $m \in \N$. Then
\[\sum_{r = m}^\infty r^{-2} \le \frac{2}{m}.\]
\end{lemma}

\begin{proof}
We calculate
\[\sum_{r = m}^\infty r^{-2} \le \sum_{r = m}^\infty \frac{2}{r(r+1)} = 2 \sum_{r = m}^\infty \left(\frac{1}{r} - \frac{1}{r+1}\right) = \frac{2}{m}.\]
\end{proof}

We are now ready to prove Lemma~\ref{lem:combinatorial}, which we restate below for the reader's convenience.


\combinatorial*

\begin{proof}
Define a universal constant
\[C = \sum_{s = 1}^\infty \frac{1}{s \log^2 (2s)}.\]
Note that $C < \infty$ by, say, the Cauchy condensation test. 

We begin with a simple, but important, structural observation about the set $P$. Let $t = (t_1, \dots, t_R) \in [k]_0^R$ be a sequence such that $t_1 + \dots + t_R > N$. Let $M = \lfloor \frac{N}{2k} \rfloor$. Then  we claim that there exists an $s \in \{M \dots, R\}$ such that $t_i \ge N / (2C s \log^2 (2s))$ for at least $s$ indices $i \in [R]$. To see this, assume without loss of generality that the entries of $\vec{t}$ are sorted so that $t_1 \ge t_2 \ge \dots \ge t_R$. Then there must exist an $s \geq M$ such that $t_s \ge N /(2Cs \log^2 (2s))$. Otherwise, because no $t_i$ can exceed $k$, we would have:
\[t_1 + \dots + t_R  < M \cdot k + \sum_{s = 1}^\infty \frac{N}{2C s \log^2 (2s)} \le \frac{N}{2} + \frac{N}{2C} \cdot \sum_{s = 1}^\infty\frac{1}{s \log^2 (2s)} = \frac{N}{2}+ \frac{N}{2} = N.\]
Since the entries of $\vec{t}$ are sorted, the preceding values $t_1, \dots, t_{s-1} \geq N / (2Cs \log^2 (2s))$ as well.

For each subset $S \subseteq [R]$, define
\[P_S = \{\vec{t} \in P : t_i \ge N / \left(2C |S| \log^2 (2|S|)\right)\text{ for all indices }i \in S\}.\]
The observations above guarantee that for every $\vec{t}=(t_1, \dots, t_R) \in P$, 
there exists some set $S$ of size at least $s \in \{M, \dots, R\}$ such that 
$t_i \geq N / (2Cs \log^2 (2s))$ for all $i \in S$. Hence,
\begin{align*}
\sum_{\vec{t} \in P} \prod_{i = 1}^R \eta_i(t_i)  & \leq \sum_{s = M}^{R} \sum_{S \subseteq [R] : |S| = s} \sum_{\vec{t} \in P_S} \prod_{i = 1}^R \eta_i(t_i) \\
& \leq \sum_{s = M}^{R} \binom{R}{s} \left( \sum_{r = \lceil N / \left(2Cs\log^2 (2s)\right)\rceil}^k \eta_i(r) \right)^s \left( \sum_{r = 0}^k \eta_i(r) \right)^{R-s} \\
&\leq  2^{-R}\sum_{s = M}^{R} \binom{R}{s} \left( \sum_{r = \lceil N / \left(2Cs \log^2 (2s)\right) \rceil}^k 10 (r+1)^{-2} \right)^s & \hspace{-18mm} \text{by Properties } \eqref{eqn:omega-bound} \text{ and } \eqref{eqn:omega-decay}\\
&\le 2^{-R} \sum_{s = M}^{R} \left( \frac{Re}{s} \right)^s \left( \frac{40 C s\log^2 (2s)}{N} \right)^s & \hspace{-18mm} \text{by Lemmas \ref{lem:binomial} and \ref{lem:inv-sq-tail}}\\
&\le 2^{-R}  \sum_{s = M}^{R} 4^{-s} &\hspace{-18mm}\text{setting } N = \lceil (160 C e) \cdot R \log^2(2R) \rceil  \\
&\le 2^{-R} \cdot 2^{-M}.
\end{align*}

The claim follows as long as $M = \lfloor \frac{N}{2k} \rfloor \ge 2 \log (2NR) \cdot R / k$, which is true for the setting of $N$ chosen above.
\end{proof}


\section{Applications}
\label{sec:applications}
\label{sec:apps}
\subsection{Approximate Rank and Quantum Communication Complexity of AC$^0$}
For a matrix $F \in \bits^{N \times N}$, 
the $\eps$-approximate rank of $F$, denoted $\text{rank}_{\eps}(F)$, is the least rank
of a matrix $A \in \R^{N \times N}$ such that $|A_{ij} - F_{ij}| \leq \eps$ for all $(i, j) \in [N] \times [N]$.
Sherstov's pattern matrix method \cite{patmat} allows one to translate approximate degree lower bounds
into approximate rank lower bounds in a black-box manner.
Moreover, the logarithm 
of the approximate rank of a communication matrix is known to lower bound its quantum communication complexity, even when prior
entanglement is allowed \cite{ls09rank}. 
By combining the pattern matrix method
with Theorems \ref{thm:acz} and \ref{thm:dnf}, we obtain the following corollary.

\begin{corollary}
For any constant $\delta>0$, there is an $\operatorname{AC}^0$ function $F \colon \bits^n \times \bits^n \to \bits$ such that $[F(x, y)]_{x, y}$
has approximate rank $\text{rank}_{1/3}(F) \geq \exp(n^{1-\delta})$. 
Similarly, there is a DNF $F \colon \bits^n \times \bits^n \to \bits$ of width $\polylog(n)$ (and quasipolynomial size)
such that $[F(x, y)]_{x, y}$
has approximate rank at least $\exp(n^{1-\delta})$. 
Moreover, the quantum communication complexity of $F$ (with arbitrary prior entanglement), denoted $Q^*_{1/3}(F)$, is $\Omega(n^{1-\delta})$. 
\end{corollary}
\begin{proof}
Let $f$ be the AC$^0$ function, or low-width DNF, with $(1/2)$-approximate degree at least $n^{1-\delta}$ whose existence is guaranteed
by Theorem \ref{thm:acz} or Theorem \ref{thm:dnf} respectively. The pattern matrix method \cite[Theorem 8.1]{patmat} implies that
 the function 
$F \colon \bits^{4n}\times \bits^{4n} \to \bits$ given by
$$F(x, y) = f\left(\dots, \vee_{j=1}^4 \left(x_{i, j} \wedge y_{i, j}\right) \dots\right)$$
satisfies
 $\text{rank}_{1/3}(F) \geq \exp(\Omega(n^{1-\delta}))$. 
Moreover, if $f$ is computed by a Boolean circuit of depth $k$ and polynomial size,
then $F$ is computed by a Boolean circuit of polynomial size and depth $k+2$. 
Similarly, if $f$ is computed by a DNF formula of width $w$, then $F$ is computed
by a DNF formula of width $O(w)$.
The claimed lower bound on $Q^*_{1/3}(F)$ follows from the fact that for any $2^n \times 2^n$ matrix $F$, we have
$Q^*_{1/3}(F) \geq \Omega( \log \text{rank}_{1/3}(F)) - O(\log n)$ \cite{ls09rank}. 
\end{proof}

The best previous lower bound on the approximate rank and quantum communication complexity
of an AC$^0$ function was $\exp\left(\tilde{\Omega}(n^{2/3})\right)$ and $\tilde{\Omega}(n^{2/3})$ respectively. This follows
from combining the Element Distinctness lower bound (Theorem \ref{thm:ed}),
with the pattern matrix method \cite{patmat}.

\subsection{Nearly Optimal Separation Between Certificate Complexity and Approximate Degree}
Certificate complexity, approximate degree, Fourier degree, block sensitivity, and deterministic, randomized, and quantum query complexities are all natural measures of the complexity of Boolean functions, with many applications in theoretical computer science. While it is known that 
all of these measures are polynomially related, much effort has been devoted
to understanding the maximal possible separations between these measures. Ambainis et al. \cite{ambainisetal}, building
on techniques of G{\"{o}}{\"{o}}s, Pitassi, and Watson \cite{gpw}, 
recently made remarkable progress in this direction, establishing a number of surprising separations
between several of these measures. Subsequent work by Aaronson, Ben-David, and Kothari \cite{cheatsheets} unified and strengthened a number of these separations.

Let $f \colon \bits^{n} \to \bits$ be a (total) Boolean function. In this section, we study the relationship between certificate complexity, denoted $C(f)$ and defined below, and approximate degree. We build on Theorem \ref{thm:dnf} to construct a function $F \colon \bits^n \rightarrow \bits$ with $\adeg(F) = n^{1-o(1)}$ and certificate complexity $n^{1/2 + o(1)}$. The function $F$ exhibits what is essentially the maximal possible separation between these two measures, as it is known that $\adeg(f) = O(C(f)^2)$ for all Boolean functions $f$.\footnote{This follows  by combining the relationship $D(f) \leq C(f) \cdot \text{bs}(f)$ \cite{qqc1} with the relationships $\text{bs}(f) \leq C(f)$ and $\adeg(f) \leq D(f)$.}
The best previous separation was reported by Aaronson et al. \cite{cheatsheets}, who gave a 
function $f$ with $\adeg(f) = \tilde{\Omega}(C(f)^{7/6})$.

\begin{theorem} \label{thm:adegcert-separation}
There is a Boolean function $F: \bits^n \to \bits$ such that $\adeg(F) \geq C(F)^{2-o(1)}$.
\end{theorem}

\paragraph{Certificate complexity definitions.} \label{sec:certdef}
Let $f \colon \bits^n \rightarrow \bits$, and let $x \in \bits^n$.
A subset $S \subseteq \{1, \dots, n\}$
is a \emph{$(-1)$-certificate (respectively, $(+1)$-certificate) for $f$ at $x$} if for all inputs $y \in \bits^n$ such that $y_i = x_i$ for all $i \in S$, it holds that $f(y)=f(x)=-1$ (respectively, $f(y)=f(x)=1$). 
For any $x \in \bits^n$, let $C(f, x)$ denote the minimum size of a certificate for $f$ at $x$.
Define $C(f) := \max_{x \in \bits^n} C(f, x)$. Define the $(-1)$-certificate complexity of $f$
to be $C_{-1}(f) := \max_{x \in f^{-1}(-1)} C(f, x)$, and the $(+1)$-certificate complexity of $f$ to be 
 $C_{+1}(f) := \max_{x \in f^{-1}(+1)} C(f, x)$.


 \subsubsection{Warm-Up: A Power $3/2$ Separation}
 
 Before proving Theorem~\ref{thm:adegcert-separation}, we begin by proving a weaker separation that illustrates most of the ideas in our construction.
 
 \begin{proposition}\label{thm:warmup}
There is a Boolean function $F \colon \bits^n \to \bits$ such that $\adeg(F) \geq C(F)^{3/2-o(1)}$.
\end{proposition}

\begin{proof}
 
While Theorem \ref{thm:dnf} is stated only for constant $k\geq 1$,
 the proof is easily seen to hold when $k$ is a function of $n$. 
 In particular, explicitly accounting for the constant factor loss that occurs in each step of the inductive proof of Theorem~\ref{thm:dnf}, we obtain the following statement that holds even if $k$ grows with $n$.
 \begin{theorem}[Generalized Version of Theorem \ref{thm:dnf}] \label{thm:gendnf}
 For any integer $k \geq 1$, there is an (explicitly given) monotone DNF on $n \cdot \log^{4k-4}(n)$ variables of width $O(\log^{2k-1}(n))$ that computes a function with approximate degree $2^{-O(k)} \cdot n^{1-2^{k-1}/3^k} \cdot \log^{3-2^{k+2}/3^{k}}(n)$.
 \end{theorem}
 Applying Theorem \ref{thm:gendnf} for an appropriately chosen $k=O(\log\log n)$ yields a function $f : \bits^M \to \bits$
 that is computed by a DNF on $M  \leq n \cdot \log^{O(k)}(n) \leq n^{1+o(1)}$ variables
 with width $O(\log^{O(k)}(n))\leq M^{o(1)}$. Equivalently, $C_{-1}(f) \leq M^{o(1)}$.
 Moreover, $\adeg(f) \geq n^{1-o(1)} \geq M^{1-o(1)}$.

 
 
 Let $\MAJ_{10 \log M}$ denote the Majority function on $10 \log M$ bits. 
The following result is implicit in \cite[Theorem 4.2]{bchtv}.

 \begin{lemma}[Bouland et al. \cite{bchtv}] \label{lem:maj-amp}
 Let $f  \colon \bits^M \to \bits$.
 Then $$\adeg_{\eps}(\MAJ_{10 \log M} \circ f) \geq \adeg(f)$$ for 
 $\eps = 1-1/M^2$. 
 \end{lemma}
 
 Now consider the block-composed function $F = \AND_M \circ \MAJ_{10 \log M} \circ f$.
This is a function on $10M^2 \log M$ variables, and Lemma \ref{lem:maj-amp} together with Proposition~\ref{prop:sherstov-degree}
implies that \begin{equation}\label{loser1} \adeg(F) \geq M^{1/2} \cdot \adeg(f) \geq 
M^{3/2-o(1)}. \end{equation}

We now show that the function $F$ has certificate complexity $C(F) \leq M^{1+o(1)}$. 
Let $\hat{F} = \MAJ_{10 \log M} \circ f$.
Then $C_{-1}(\hat{F}) \leq 5\log M \cdot C_{-1}(f) = M^{o(1)}$; this uses the fact
that in order to certify that $\MAJ_{10 \log M}$ evaluates to $-1$, it is enough 
to certify that at least half of its inputs are equal to $-1$.
And, trivially, $C_{+1}(\hat{F}) \leq 10M \log M$.

Any input $z$ to $\AND_M$
has a certificate $S$ such that $z_i = +1$ for at most one index $i \in S$.
By composing certificates, it follows that \begin{equation} \label{loser} C(F) \le C_{+1}(\hat{F}) + M \cdot C_{-1}(\hat{F}) \leq M^{1+o(1)}. \end{equation}
Combining \eqref{loser1} and \eqref{loser} completes the proof of Proposition~\ref{thm:warmup}.

\end{proof}

 \subsubsection{A Nearly Quadratic Separation}
 To improve Proposition \ref{thm:warmup} to a nearly quadratic separation,
we replace $\AND_M$ in the definition of $F$ with a function $f^*$ defined on roughly $M$ variables, such that $\adeg(f^*) \geq M^{1-o(1)}$.  This function $f^*$ must moreover possess certificates satisfying the same key property as the $\AND$ function.
Namely, every input $z$ to $f^*$ must have a certificate $S$ such that $z_i = +1$ for only a small (as we will see, $M^{o(1)}$) number of indices $i \in S$.
 
\paragraph{Construction of $f^*$.}
While the function $f \colon \bits^M \to \bits$ considered in the proof of Proposition~\ref{thm:warmup} satisfies the requisite approximate degree bound, it lacks the key property regarding its certificates. Hence, we must modify the function $f$ to obtain a suitable function
$f^*$. The modification we use generalizes a technique introduced by Aaronson et al. \cite[Theorem 8]{cheatsheets} to give separations between quantum query complexity and certificate complexity, Fourier degree, and approximate degree.

%


\begin{definition} \label{def:hstar}
Let $f \colon \bits^{n} \to \bits$ be computed by a DNF formula $\mathcal{C}_f$ of width $w$. We define a function $f^* \colon \bits^{2n} \to \bits$ as follows.
Let each of the first $n$ inputs to $f^*$ be associated with an input to $f$, and each of the last $n$ inputs of $f^*$ be associated with the negation of an input to $f$. 
For an $i \in [n]$, let $(x_i, x_{n+i})$ be the pair of inputs
to $f^*$ corresponding to the $i$th input to $f$, and say that the pair is \emph{balanced}
if exactly one of $x_i, x_{n+i}$ is equal to $-1$ (and exactly one is equal to $+1$).

For an input $x \in \bits^{2n}$, define $\gamma(x) \in \{-1, +1, \perp\}^{n}$
by $$(\gamma(x))_i = \begin{cases} -1 \text{ if } (x_i, x_{n+i}) \text{ is balanced and } x_i=-1, \\
+1  \text{ if } (x_i, x_{n+i}) \text{ is balanced and } x_{n+i}=-1,\\
\perp \text{ otherwise.}\end{cases}$$

We say a clause of the DNF formula $\mathcal{C}_f$ is \emph{satisfied} by
a vector $y \in \{-1, +1, \perp\}^{n}$
if every literal in that clause is satisfied by $y$ (if $y_i=\perp$, then any literal corresponding to an input $i$
or its negation is automatically unsatisfied).
Define $f^* \colon \bits^{2n} \to \bits$ by:
 \begin{equation*} f^*(x)  = \begin{cases}  
 -1 \text{ if there is a clause of } \mathcal{C}_f \text{ that is satisfied by } \gamma(x), \text{ and for all } i \in [n], (x_i, x_{n+i}) \neq (+1, +1).\\
 +1 \text{ otherwise.} 
\end{cases}\end{equation*}
\end{definition}

Let $\eps > 0$ and suppose $p : \bits^{2n} \to \mathbb{R}$ is a polynomial with $|p(x) - f^*(x)| \le \eps$ for every $x \in \bits^{2n}$. Then the polynomial $q: \bits^n \to \mathbb{R}$ defined by $q(y) = p(y_1, \dots, y_n, -y_1, \dots, -y_n)$ satisfies $|q(y) - f(y)| \le \eps$ for all $y \in \bits^n$, since the definition of $f^*$ guarantees that $f(y) = f^*(y_1, \dots, y_n, -y_1, \dots, -y_n)$. 
Hence, for every $\eps > 0$, we have
\begin{equation}\label{sighhstar} \adeg_{\eps}(f^*) \ge \adeg_{\eps}(f).
\end{equation} 


\paragraph{Completing the Proof of Theorem~\ref{thm:adegcert-separation}.}
Let $f^*$ denote the function obtained by applying Definition \ref{def:hstar} to 
the function $f \colon \bits^M \to \bits$ described in the proof of Proposition~\ref{thm:warmup}. Recall that $f$ is computed by a DNF $\mathcal{C}_f$ 
of width $M^{o(1)}$.
Hence $f^*$ is a function on $2M$ variables,
and by Inequality \eqref{sighhstar}, 
$$\adeg(f^*) \geq \adeg(f) \geq M^{1-o(1)}.$$

We now argue that every input $x$ to $f^*$ has a certificate $S$ in which at most $M^{o(1)}$
entries of $x|_S$ are equal to $+1$. 
To see this, first let $x$ be any input in $\left(f^*\right)^{-1}(-1)$. Then by definition of $f^*$, it suffices
to certify that (a) there is a clause of $\mathcal{C}_f$ that is satisfied by $\gamma(x)$
and (b) there is no $i \in [M]$ such that $(x_i, x_{M+i})=(+1, +1)$. Letting $w$
denote the width of $\mathcal{C}_f$, condition (a) can be certified  
by providing at most $2 w=M^{o(1)}$ indices of $x$, and condition (b)
can be certified by supplying all of the coordinates of $x$ that are equal to $-1$.

Now suppose that $x$ is an input in  $(f^*)^{-1}(+1)$. There are two kinds of such inputs to certify. 
The first kind is any input $x$ with $(x_i, x_{M+i})=(+1, +1)$ for some $i \in [M]$. Such an input can be certified by providing 
$x_i$ and $x_{M+i}$.
The second kind is an input such that for all $i \in [M]$, the pair $(x_i, x_{M+i})$ has at least one coordinate
equal to $-1$, yet no clause of $\mathcal{C}_f$ is satisfied by $\gamma(x)$.  
This kind of input can be certified by providing the indices of all $(-1)$'s in the input $x$. Such a certificate is enough to reveal $(\gamma(x))_i$ for all $i$ under the assumption that every pair $(x_i, x_{M+i})$
with exactly one $(-1)$ provided is balanced. While this certificate does not prove that every pair $(x_i, x_{M+i})$ is actually balanced, it is still enough to prove that there is no clause of $\mathcal{C}_f$ that is satisfied. This is because changing a pair $(x_i, x_{M+i})$ from balanced to unbalanced cannot
cause an unsatisfied clause of $\mathcal{C}_f$ to become satisfied.

To summarize, the value of $f^*(x)$ can always be certified by providing
at most $M^{o(1)}$ indices of $x$ that are equal to $+1$.
To complete our construction, let $F = f^* \circ \hat{F}$, where $\hat{F} =\MAJ_{10 \log M} \circ f$. This is a function on $M^{2+o(1)}$ variables.
With the aforementioned property of the certificates of $f^*$ in hand, the argument that $\adeg(F) \geq C(F)^{2-o(1)}$
is identical to that of the previous section. Indeed,
by composing certificates for $f^*$ and $\hat{F}$, one obtains 
\begin{equation} \label{finish} C(F) \leq M^{o(1)}\cdot C_{+1}(\hat{F}) + M \cdot C_{-1}(\hat{F}) \leq M^{1+o(1)}.\end{equation}
Since $\adeg(f^*) \geq M^{1-o(1)}$, 
Lemma \ref{lem:maj-amp} and Proposition~\ref{prop:sherstov-degree} imply that
\begin{equation} \label{done} \adeg(F) \geq M^{2-o(1)}.\end{equation}
Combining Equations  \eqref{finish} and \eqref{done} completes the proof of Theorem~\ref{thm:adegcert-separation}.

\subsection{Secret Sharing Schemes}
Bogdanov et al. \cite{viola} observed that for any
$f \colon \bits^n \to \bits$ and integer $d>0$, any dual polynomial $\mu$
for the fact that $\adeg_{\eps}(f) \geq d$ leads to a scheme
for sharing a single secret bit $b \in \bits$ among $n$ parties as follows.
Decompose $\mu$ as $\mu_{+} - \mu_{-}$, where $\mu_{+}$ and $\mu_{-}$
are non-negative functions with $\|\mu_{+}\|_1 = \|\mu_{0}\|_1 = 1/2$. Then in order to split $b$
among $n$ parties,
one draws an input $x=(x_1, \dots, x_n) \in \bits^n$ from the distribution $2 \cdot \mu_b$, and
gives bit $x_i$ to the $i$th party.
In order to reconstruct $b$, one simply applies $f$ to $(x_1, \dots, x_n)$.

Because $\mu$ is $\eps$-correlated with $f$, the probability of correct reconstruction if the bit is 
chosen at random
is at least $(1+\eps)/2$ (and the the \emph{reconstruction advantage}, 
defined to equal $\Pr_{x \sim \mu_{+}}[f(x) = 1] - \Pr_{x\sim\mu_{-}}[f(x) = 1]$, is at least
$\eps$). 
The fact that $\mu$ has pure high degree at least $d$
means that any subset of shares of size less than $d$ provides no information 
about the secret bit $b$. We direct the interested reader to \cite{viola} for further details.

Hence, an immediate corollary of our new approximate degree
lower bounds for AC$^0$ is the following.

\begin{corollary} \label{cor:secret-sharing}
For any arbitrarily small constant $\delta > 0$, there is a secret sharing scheme 
that shares a single bit $b$ among $n$ parties by assigning a bit $x_i$ to each party $i$. The scheme satisfies the following properties.
\begin{enumerate}[(a)]
\item The reconstruction procedure is computed by an $\operatorname{AC}^0$ circuit.
\item The reconstruction advantage is at least 0.49.
\item Any subset of shares of size less than $d = \Omega(n^{1-\delta})$ provides no information 
about the secret bit $b$.
\end{enumerate}
\end{corollary}

The above corollary improves over an analogous result of Bogdanov et al. \cite{viola},
who used the Element Distinctness lower bound (cf. Theorem \ref{thm:ed})
to  give a scheme for which subsets of shares of size less than $d = \Omega(n^{2/3})$ provides no information 
about the secret bit $b$.

To make the secret sharing scheme of Corollary \ref{cor:secret-sharing} explicit, 
one needs an explicit dual polynomial witnessing our new approximate degree lower 
bounds for AC$^0$ (cf. Theorem \ref{thm:acz}). Strictly speaking, our proof of Theorem \ref{thm:acz}
does not achieve this, owing to the \emph{primal-based} symmetrization step of Section~\ref{sec:step1}.
However, this issue is easily addressed. 

In more detail,
recall that Theorem \ref{thm:main-amp} establishes that
there is an AC$^0$ function 
$$\promcomp\colon \{-1, 1\}^{N'}_{\le N} \to \bits$$ for some $N' > N$, such that $\promcomp$ has approximate degree 
at least $\Omega(N^{1-\delta})$. 
In fact, the proof of Theorem \ref{thm:main-amp} constructs an explicit dual polynomial $\psi$ witnessing this approximate degree bound.
Definition \ref{def:g} and Corollary~\ref{cor:reduction} define an associated function $g^* \colon \bits^{m} \to \bits$, for $m=\tilde{O}(N)$,
such that $\adeg_{\eps}(g^*) \ge \adeg_{\eps}(\promcomp) \cdot \lceil \log(R+1) \rceil.$
A natural averaging construction shows how to translate the dual polynomial $\psi \colon  \{-1, 1\}^{N'}_{\le N} \to \reals$ for $\promcomp$ into a dual polynomial $\phi \colon \bits^m \to \reals$ for $g^*$. 
The analysis in the proof of Theorem \ref{thm:sym-vs-prom} can then be used to show that this transformation
preserves pure high degree, and that the correlation of $\phi$ with $g^*$
is the same as the correlation of $\psi$ with $\promcomp$. (These remarks also apply to the augmented construction of $g$ in Definition~\ref{def:h}.)

We further believe that closer inspection of this dual witness should show that shares from the resulting scheme can be \emph{sampled} by an AC$^0$ circuit. 

\section{Future Directions}

\subsection{Stronger Results for Constant Error Approximation} Throughout this section, $\delta$ denotes an arbitrarily small positive constant.
While our $\Omega(n^{1-\delta})$ lower bound on the approximate degree of AC$^0$ 
comes close to resolving Problem 1 from the introduction, we
fall short of a complete solution. Can our techniques
be refined to give an $\Omega(n)$ lower bound on the approximate degree
of a function in AC$^0$? Even the approximate degree of the
$\mathsf{SURJECTIVITY}$ function remains unresolved. It is reasonable to conjecture
that this function has essentially maximal approximate degree, $\Omega(n)$,
yet our methods do not improve on the known
$\Omega(n^{2/3})$ lower bound for this function.  

It would also be very interesting
to extend our $\Omega(n^{1-\delta})$ lower bounds for DNFs of polylogarithmic width and quasipolynomial
size to DNFs of polynomial size (and ideally of logarithmic width). Currently,
the best known lower bound on the approximate degree of polynomial size 
DNFs remains $\tilde{\Omega}(n^{2/3})$ for Element Distinctness.

For any constant integer $k > 0$, the $k$-sum function is a DNF of width $O(\log n)$ that
might have approximate degree $\Omega(n^{k/(k+1)})$ \cite{ksum, ksumadversary}. Another candidate DNF that might
have approximate degree polynomially larger than $\Omega(n^{2/3})$ 
is the
$k$-distinctness function for $k \geq 3$. (The best known upper bound on
the approximate degree of the $k$-distinctness function is $O(n^{1-2^{k-2}/(2^k-1)})$; this bound approaches $n^{3/4}$ as $k \rightarrow \infty$ \cite{kdistinctness}.)
However, we believe that substantially new techniques will be required to resolve the approximate degree 
of these specific candidates. 
As explained in Section \ref{sec:y}, our analysis is tailored to 
showing (near-)optimality of robustification-based approximating 
polynomials for the functions we consider, in a sense that can be made precise via complementary slackness. 
But the best known approximating polynomials for $k$-sum and $k$-distinctness 
are derived from sophisticated quantum algorithms \cite{ksum, kdistinctness}.
In particular, they are not constructed via robustification. 
Hence, we expect that any proof of a novel approximate degree lower bound for these functions will have 
to look very different than our own, as they will have to implicitly engage with 
non-robustification based approximating polynomials.

\subsection{Stronger Results for Large Error Approximation}
Another open direction is to strengthen our $\eps$-approximate degree lower bounds on AC$^0$
from $\eps=1/3$ to $\eps$ much closer to 1. 
For example, the following two variants
of Problem~\ref{problem:bounded} from the introduction are open.

\begin{problem} \label{problem:discrepancy}
Is there a constant-depth circuit in $n$
variables with $\eps$-approximate degree $\Omega(n)$, for (say) $\eps=1-2^{-\Omega(n)}$?
\end{problem}

\begin{problem} \label{problem:threshold}
Is there a constant-depth circuit in $n$
variables with  $\eps$-approximate degree $\Omega(n)$, for \emph{any} $\eps< 1$?
\end{problem}


Problem~\ref{problem:threshold} is equivalent to asking whether there is an AC$^0$ function
with linear \emph{threshold degree}. 
Resolving Problems~\ref{problem:discrepancy} and~\ref{problem:threshold} would have a wide variety of consequences in 
computational learning theory, circuit complexity, and communication complexity 
(see, e.g., \cite{bchtv, sherstov14, btdl} and the references therein).

Despite attention by many researchers, the best known lower bounds in the directions of Problems 3 and 4 are:
\begin{enumerate}[(a)]
\item For any constant $\Gamma > 0$, a depth-3 circuit
with $\eps$-approximate degree $\Omega(n^{1/2-\delta})$ for $\eps=1-2^{-n^{\Gamma}}$ \cite{btdl},
\item A depth-3 circuit with threshold degree $\Omega(n^{3/7})$ \cite{sherstov15}, and
\item A depth-4 circuit with threshold degree $\Omega(n^{1/2})$ \cite{sherstov15}.
\end{enumerate}
We believe that the following three results in the directions of Problems~\ref{problem:discrepancy} and~\ref{problem:threshold} should be achievable 
via relatively modest  extensions of our techniques. 

\medskip
First, it should be possible to nearly resolve Problem~\ref{problem:discrepancy} as follows.
Recall from Section \ref{sec:additionalpriorwork} that 
our recent work \cite{btdl} also 
proved stronger hardness amplification results for approximate degree
by moving beyond block composed functions. The methods of \cite{btdl}
amplify
approximation error but not degree, while in this paper we amplify degree but not approximation error.
We believe that it is possible to combine the two sets of techniques to exhibit a function in AC$^0$
on $n$ variables with $\eps$-approximate degree at least $n^{1-\delta}$, even for
$\eps=1-2^{-\Omega(n^{1-\delta})}$. Such a result would translate in a black-box manner into lower bounds
of $2^{\Omega(n^{1-\delta})}$ on the margin complexity, (multiplicative inverse of) discrepancy, threshold weight,
and Majority-of-Threshold circuit size of AC$^0$, nearly matching trivial $2^{O(n)}$ upper bounds.



%

\medskip
Second,
we are confident that the polylogarithmic width DNF $f \colon \bits^n \to \bits$ of approximate degree $\Omega(n^{1-\delta})$
exhibited in Theorem \ref{thm:dnf} in fact has large \emph{one-sided} approximate degree \cite{bt14}.
Moreover, this should be provable via
a modest extension of our techniques.
Combining such a lower bound with a result of Sherstov \cite{sherstov14} would imply
that $\AND_{n^{1-\delta}} \circ f$ has threshold degree $\Omega(n^{1-\delta})$,
thereby yielding a depth three circuit (of quasipolynomial size) on $N = n^{2-2\delta}$ variables 
with threshold degree $\Omega(N^{1/2-\delta})$. 

\medskip
Third, we believe that the following function $g$ on $O(n \log^4 n)$ variables 
has threshold degree $\Omega(n^{3/5})$.
Let $f_n = \AND_{n^{1/5}} \circ \OR_{n^{2/5}} \circ \AND_{n^{2/5}}$, and let $g$
be the harder function obtained by applying the construction of Theorem \ref{thm:main} to $f_n$. 
Note that $g$ is computed by a circuit of depth 5.

Sherstov \cite{sherstov14} constructed a dual polynomial $\psi$
witnessing the fact that $$\deg_{\pm}\left(\AND_{n^{1/5}} \circ \OR_{n^{2/5}} \circ \AND_{n^{2/5}} \circ \OR_{n^{2/5}}\right) = \Omega(n^{3/5}).$$ 
(Note that this block composed function is defined over $n^{7/5}$ variables.) In order to show that $g$ 
likewise has threshold degree $\Omega(n^{3/5})$, our results from Section \ref{sec:step1} imply that  
it is enough to
``zero out'' the mass that $\psi$ places on inputs of Hamming weight larger than a suitable
threshold $N=\tilde{O}(n)$, without affecting the sign of $\psi$ on the remaining inputs.
We believe that is possible to achieve this via a refinement of the zeroing technique used in this work.

\paragraph{A final ambitious direction.}
A more ambitious direction toward resolving Problems~\ref{problem:discrepancy} and~\ref{problem:threshold} would be to obtain a version of our hardness amplification result (Theorem \ref{thm:main}) that
(a) applies to threshold degree rather than approximate degree and (b) can be applied recursively. This would allow
one to obtain an $\Omega(n^{1-\delta})$
lower bound on the threshold degree of AC$^0$, nearly resolving Problem~\ref{problem:threshold} above. 

One might hope to obtain such a result by extending the above envisioned analysis for obtaining an $\Omega(n^{3/5})$
threshold degree lower bound, so as to allow recursive application of the construction and analysis.
However, we believe that achieving this goal will require substantial new ideas.
The only available techniques for recursively amplifying threshold degree bounds are due to Sherstov \cite{sherstov14, sherstov15}, who considers block composed functions of the form $\mathsf{OR} \circ f$. Specifically, he uses a dual witness for the \emph{outer function} $\OR$ to ``amplify the efficacy'' of a dual witness for the \emph{inner function} $f$.

In contrast, our recursive construction in this paper considers block compositions of the form $f \circ \OR$, and uses a dual witness for 
the \emph{inner function} $\OR$ to ``amplify the efficacy'' of a dual witness for the \emph{outer function} $f$. This difference appears to 
prevent us from combining the methods of \cite{sherstov14} with our own in a manner that would enable recursive application. 
Finding a way to reconcile the two approaches may present a promising avenue for obtaining a (near-)resolution of Problem~\ref{problem:threshold}.

\medskip \noindent \textbf{Acknowledgements.} We are grateful to Shalev Ben-David for illuminating conversations regarding separations between approximate degree and certificate complexity, and to Robin Kothari and Sasha Sherstov for valuable comments on an earlier version of this manuscript. 

\bibliographystyle{plain}
\bibliography{mainsad}

\begin{thebibliography}{10}

\bibitem{focstutorial}
Scott Aaronson.
\newblock The polynomial method in quantum and classical computing.
\newblock In {\em 49th Annual {IEEE} Symposium on Foundations of Computer
  Science, {FOCS} 2008, October 25-28, 2008, Philadelphia, PA, {USA}}, page~3,
  2008.

\bibitem{qqc2}
Scott Aaronson.
\newblock Impossibility of succinct quantum proofs for collision-freeness.
\newblock {\em Quantum Information {\&} Computation}, 12(1-2):21--28, 2012.

\bibitem{cheatsheets}
Scott Aaronson, Shalev Ben{-}David, and Robin Kothari.
\newblock Separations in query complexity using cheat sheets.
\newblock In {\em Proceedings of the 48th Annual {ACM} {SIGACT} Symposium on
  Theory of Computing, {STOC} 2016, Cambridge, MA, USA, June 18-21, 2016},
  pages 863--876, 2016.

\bibitem{aaronsonshi}
Scott Aaronson and Yaoyun Shi.
\newblock Quantum lower bounds for the collision and the element distinctness
  problems.
\newblock {\em J. ACM}, 51(4):595--605, 2004.

\bibitem{ambainis}
Andris Ambainis.
\newblock Polynomial degree and lower bounds in quantum complexity: Collision
  and element distinctness with small range.
\newblock {\em Theory of Computing}, 1(1):37--46, 2005.

\bibitem{ksum}
Andris Ambainis.
\newblock Quantum walk algorithm for element distinctness.
\newblock {\em {SIAM} J. Comput.}, 37(1):210--239, 2007.

\bibitem{ambainisetal}
Andris Ambainis, Kaspars Balodis, Aleksandrs Belovs, Troy Lee, Miklos Santha,
  and Juris Smotrovs.
\newblock Separations in query complexity based on pointer functions.
\newblock In {\em Proceedings of the 48th Annual {ACM} {SIGACT} Symposium on
  Theory of Computing, {STOC} 2016, Cambridge, MA, USA, June 18-21, 2016},
  pages 800--813, 2016.

\bibitem{readonceformulae}
Andris Ambainis, Andrew~M. Childs, Ben Reichardt, Robert Spalek, and Shengyu
  Zhang.
\newblock Any and-or formula of size n can be evaluated in time
  n$^{\mbox{1/2+o(1)}}$ on a quantum computer.
\newblock {\em SIAM J. Comput.}, 39(6):2513--2530, 2010.

\bibitem{qqc1}
Robert Beals, Harry Buhrman, Richard Cleve, Michele Mosca, and Ronald de~Wolf.
\newblock Quantum lower bounds by polynomials.
\newblock {\em J. {ACM}}, 48(4):778--797, 2001.

\bibitem{beame}
Paul Beame and Widad Machmouchi.
\newblock The quantum query complexity of {AC}\({}^{\mbox{0}}\).
\newblock {\em Quantum Information {\&} Computation}, 12(7-8):670--676, 2012.

\bibitem{beigel}
Richard Beigel.
\newblock Perceptrons, {PP}, and the {P}olynomial {H}ierarchy.
\newblock {\em Computational Complexity}, 4:339--349, 1994.

\bibitem{kdistinctness}
Aleksandrs Belovs.
\newblock Learning-graph-based quantum algorithm for k-distinctness.
\newblock In {\em 53rd Annual {IEEE} Symposium on Foundations of Computer
  Science, {FOCS} 2012, New Brunswick, NJ, USA, October 20-23, 2012}, pages
  207--216. {IEEE} Computer Society, 2012.

\bibitem{ksumadversary}
Aleksandrs Belovs and Robert Spalek.
\newblock Adversary lower bound for the k-sum problem.
\newblock In Robert~D. Kleinberg, editor, {\em Innovations in Theoretical
  Computer Science, {ITCS} '13, Berkeley, CA, USA, January 9-12, 2013}, pages
  323--328. {ACM}, 2013.

\bibitem{viola}
Andrej Bogdanov, Yuval Ishai, Emanuele Viola, and Christopher Williamson.
\newblock Bounded indistinguishability and the complexity of recovering
  secrets.
\newblock In Matthew Robshaw and Jonathan Katz, editors, {\em Advances in
  Cryptology - {CRYPTO} 2016 - 36th Annual International Cryptology Conference,
  Santa Barbara, CA, USA, August 14-18, 2016, Proceedings, Part {III}}, volume
  9816 of {\em Lecture Notes in Computer Science}, pages 593--618. Springer,
  2016.

\bibitem{bchtv}
Adam Bouland, Lijie Chen, Dhiraj Holden, Justin Thaler, and Prashant~Nalini
  Vasudevan.
\newblock On {SZK} and {PP}.
\newblock {\em Electronic Colloquium on Computational Complexity {(ECCC)}},
  23:140, 2016.

\bibitem{bvdw}
Harry Buhrman, Nikolai~K. Vereshchagin, and Ronald de~Wolf.
\newblock On computation and communication with small bias.
\newblock In {\em 22nd Annual {IEEE} Conference on Computational Complexity
  {(CCC} 2007), 13-16 June 2007, San Diego, California, {USA}}, pages 24--32.
  {IEEE} Computer Society, 2007.

\bibitem{bt13}
Mark Bun and Justin Thaler.
\newblock Dual lower bounds for approximate degree and markov-bernstein
  inequalities.
\newblock In Fedor~V. Fomin, Rusins Freivalds, Marta~Z. Kwiatkowska, and David
  Peleg, editors, {\em ICALP (1)}, volume 7965 of {\em Lecture Notes in
  Computer Science}, pages 303--314. Springer, 2013.

\bibitem{bt14}
Mark Bun and Justin Thaler.
\newblock Hardness amplification and the approximate degree of constant-depth
  circuits.
\newblock In Magn{\'{u}}s~M. Halld{\'{o}}rsson, Kazuo Iwama, Naoki Kobayashi,
  and Bettina Speckmann, editors, {\em Automata, Languages, and Programming -
  42nd International Colloquium, {ICALP} 2015, Kyoto, Japan, July 6-10, 2015,
  Proceedings, Part {I}}, volume 9134 of {\em Lecture Notes in Computer
  Science}, pages 268--280. Springer, 2015.
\newblock Full version available at
  \texttt{http://eccc.hpi-web.de/report/2013/151}.

\bibitem{btdl}
Mark Bun and Justin Thaler.
\newblock Approximate degree and the complexity of depth three circuits.
\newblock {\em Electronic Colloquium on Computational Complexity {(ECCC)}},
  23:121, 2016.

\bibitem{bttoc}
Mark Bun and Justin Thaler.
\newblock Dual polynomials for {C}ollision and {E}lement {D}istinctness.
\newblock {\em Theory of Computing}, 12(16):1--34, 2016.

\bibitem{bt16}
Mark Bun and Justin Thaler.
\newblock Improved bounds on the sign-rank of {AC}$^0$.
\newblock In Ioannis Chatzigiannakis, Michael Mitzenmacher, Yuval Rabani, and
  Davide Sangiorgi, editors, {\em 43rd International Colloquium on Automata,
  Languages, and Programming, {ICALP} 2016, July 11-15, 2016, Rome, Italy},
  volume~55 of {\em LIPIcs}, pages 37:1--37:14. Schloss Dagstuhl -
  Leibniz-Zentrum fuer Informatik, 2016.

\bibitem{difpriv2}
Karthekeyan Chandrasekaran, Justin Thaler, Jonathan Ullman, and Andrew Wan.
\newblock Faster private release of marginals on small databases.
\newblock {\em CoRR}, abs/1304.3754, 2013.

\bibitem{comm7}
Arkadev Chattopadhyay and Anil Ada.
\newblock Multiparty communication complexity of disjointness.
\newblock {\em Electronic Colloquium on Computational Complexity {(ECCC)}},
  15(002), 2008.

\bibitem{lijie1}
Lijie Chen.
\newblock Adaptivity vs. postselection, and hardness amplification for
  polynomial approximation.
\newblock In {\em 27th International Symposium on Algorithms and Computation,
  {ISAAC} 2016, December 12-14, 2016, Sydney, Australia}, pages 26:1--26:12,
  2016.

\bibitem{kothari}
Andrew~M. Childs, Shelby Kimmel, and Robin Kothari.
\newblock The quantum query complexity of read-many formulas.
\newblock In Leah Epstein and Paolo Ferragina, editors, {\em Algorithms - {ESA}
  2012 - 20th Annual European Symposium, Ljubljana, Slovenia, September 10-12,
  2012. Proceedings}, volume 7501 of {\em Lecture Notes in Computer Science},
  pages 337--348. Springer, 2012.

\bibitem{comm8}
Matei David and Toniann Pitassi.
\newblock Separating {NOF} communication complexity classes {RP} and {NP}.
\newblock {\em Electronic Colloquium on Computational Complexity {(ECCC)}},
  15(014), 2008.

\bibitem{comm6}
Matei David, Toniann Pitassi, and Emanuele Viola.
\newblock Improved separations between nondeterministic and randomized
  multiparty communication.
\newblock {\em {TOCT}}, 1(2), 2009.

\bibitem{comm2}
Dmitry Gavinsky and Alexander~A. Sherstov.
\newblock A separation of {NP} and co{NP} in multiparty communication
  complexity.
\newblock {\em Theory of Computing}, 6(1):227--245, 2010.

\bibitem{gpw}
Mika G{\"{o}}{\"{o}}s, Toniann Pitassi, and Thomas Watson.
\newblock Deterministic communication vs. partition number.
\newblock In {\em {IEEE} 56th Annual Symposium on Foundations of Computer
  Science, {FOCS} 2015, Berkeley, CA, USA, 17-20 October, 2015}, pages
  1077--1088, 2015.

\bibitem{hoyer}
Peter H{\o}yer, Michele Mosca, and Ronald de~Wolf.
\newblock Quantum search on bounded-error inputs.
\newblock In Jos C.~M. Baeten, Jan~Karel Lenstra, Joachim Parrow, and
  Gerhard~J. Woeginger, editors, {\em Automata, Languages and Programming, 30th
  International Colloquium, {ICALP} 2003, Eindhoven, The Netherlands, June 30 -
  July 4, 2003. Proceedings}, volume 2719 of {\em Lecture Notes in Computer
  Science}, pages 291--299. Springer, 2003.

\bibitem{kahn}
Jeff Kahn, Nathan Linial, and Alex Samorodnitsky.
\newblock Inclusion-exclusion: Exact and approximate.
\newblock {\em Combinatorica}, 16(4):465--477, 1996.

\bibitem{agnostic}
Adam~Tauman Kalai, Adam~R. Klivans, Yishay Mansour, and Rocco~A. Servedio.
\newblock Agnostically learning halfspaces.
\newblock {\em {SIAM} J. Comput.}, 37(6):1777--1805, 2008.

\bibitem{colt}
Varun Kanade and Justin Thaler.
\newblock Distribution-independent reliable learning.
\newblock In Maria{-}Florina Balcan and Csaba Szepesv{\'{a}}ri, editors, {\em
  Proceedings of The 27th Conference on Learning Theory, {COLT} 2014,
  Barcelona, Spain, June 13-15, 2014}, volume~35 of {\em {JMLR} Proceedings},
  pages 3--24. JMLR.org, 2014.

\bibitem{ksdnf}
Adam~R. Klivans and Rocco~A. Servedio.
\newblock Learning {DNF} in time $2^{\tilde{o}(n^{1/3})}$.
\newblock {\em J. Comput. Syst. Sci.}, 68(2):303--318, 2004.

\bibitem{klivansservedioomb}
Adam~R. Klivans and Rocco~A. Servedio.
\newblock Toward attribute efficient learning of decision lists and parities.
\newblock {\em Journal of Machine Learning Research}, 7:587--602, 2006.

\bibitem{lee}
Troy Lee.
\newblock A note on the sign degree of formulas.
\newblock {\em CoRR}, abs/0909.4607, 2009.

\bibitem{ls09rank}
Troy Lee and Adi Shraibman.
\newblock An approximation algorithm for approximation rank.
\newblock In {\em Proceedings of the 24th Annual {IEEE} Conference on
  Computational Complexity, {CCC} 2009, Paris, France, 15-18 July 2009}, pages
  351--357, 2009.

\bibitem{mp}
Marvin Minsky and Seymour Papert.
\newblock {\em Perceptrons - an introduction to computational geometry}.
\newblock {MIT} Press, 1969.

\bibitem{nisanszegedy}
Noam Nisan and Mario Szegedy.
\newblock On the degree of boolean functions as real polynomials.
\newblock {\em Computational Complexity}, 4:301--313, 1994.

\bibitem{osnewbounds}
Ryan O'Donnell and Rocco~A. Servedio.
\newblock New degree bounds for polynomial threshold functions.
\newblock {\em Combinatorica}, 30(3):327--358, 2010.

\bibitem{paturi}
Ramamohan Paturi.
\newblock On the degree of polynomials that approximate symmetric boolean
  functions (preliminary version).
\newblock In S.~Rao Kosaraju, Mike Fellows, Avi Wigderson, and John~A. Ellis,
  editors, {\em Proceedings of the 24th Annual {ACM} Symposium on Theory of
  Computing, May 4-6, 1992, Victoria, British Columbia, Canada}, pages
  468--474. {ACM}, 1992.

\bibitem{podolskii}
Vladimir~V. Podolskii.
\newblock A uniform lower bound on weights of perceptrons.
\newblock In Edward~A. Hirsch, Alexander~A. Razborov, Alexei~L. Semenov, and
  Anatol Slissenko, editors, {\em Computer Science - Theory and Applications,
  Third International Computer Science Symposium in Russia, {CSR} 2008, Moscow,
  Russia, June 7-12, 2008, Proceedings}, volume 5010 of {\em Lecture Notes in
  Computer Science}, pages 261--272. Springer, 2008.

\bibitem{comm4}
Anup Rao and Amir Yehudayoff.
\newblock Simplified lower bounds on the multiparty communication complexity of
  disjointness.
\newblock {\em Electronic Colloquium on Computational Complexity {(ECCC)}},
  21:60, 2014.

\bibitem{sherstovrazborov}
Alexander~A. Razborov and Alexander~A. Sherstov.
\newblock The sign-rank of {AC$^0$}.
\newblock {\em {SIAM} J. Comput.}, 39(5):1833--1855, 2010.

\bibitem{servediotanthaler}
Rocco~A. Servedio, Li-Yang Tan, and Justin Thaler.
\newblock Attribute-efficient learning and weight-degree tradeoffs for
  polynomial threshold functions.
\newblock In Shie Mannor, Nathan Srebro, and Robert~C. Williamson, editors,
  {\em COLT}, volume~23 of {\em JMLR Proceedings}, pages 14.1--14.19. JMLR.org,
  2012.

\bibitem{sherstov15}
A.~A. Sherstov.
\newblock The power of asymmetry in constant-depth circuits.
\newblock In {\em FOCS}, 2015.
\newblock Full version available at
  \url{http://eccc.hpi-web.de/report/2015/147/}.

\bibitem{sherstovsurvey}
Alexander~A. Sherstov.
\newblock Communication lower bounds using dual polynomials.
\newblock {\em Bulletin of the {EATCS}}, 95:59--93, 2008.

\bibitem{sherstovinclusion}
Alexander~A. Sherstov.
\newblock Approximate inclusion-exclusion for arbitrary symmetric functions.
\newblock {\em Computational Complexity}, 18(2):219--247, 2009.

\bibitem{sherstovmajmaj}
Alexander~A. Sherstov.
\newblock Separating {AC}\({}^{\mbox{0}}\) from depth-2 majority circuits.
\newblock {\em {SIAM} J. Comput.}, 38(6):2113--2129, 2009.

\bibitem{patmat}
Alexander~A. Sherstov.
\newblock The pattern matrix method.
\newblock {\em {SIAM} J. Comput.}, 40(6):1969--2000, 2011.

\bibitem{comm1}
Alexander~A. Sherstov.
\newblock The multiparty communication complexity of set disjointness.
\newblock In Howard~J. Karloff and Toniann Pitassi, editors, {\em Proceedings
  of the 44th Symposium on Theory of Computing Conference, {STOC} 2012, New
  York, NY, USA, May 19 - 22, 2012}, pages 525--548. {ACM}, 2012.

\bibitem{sherstovandor}
Alexander~A. Sherstov.
\newblock Approximating the {A}{N}{D}-{O}{R} {T}ree.
\newblock {\em Theory of Computing}, 9(20):653--663, 2013.

\bibitem{comm3}
Alexander~A. Sherstov.
\newblock Communication lower bounds using directional derivatives.
\newblock In Dan Boneh, Tim Roughgarden, and Joan Feigenbaum, editors, {\em
  Symposium on Theory of Computing Conference, STOC'13, Palo Alto, CA, USA,
  June 1-4, 2013}, pages 921--930. {ACM}, 2013.

\bibitem{sherstovhalfspaces1}
Alexander~A. Sherstov.
\newblock The intersection of two halfspaces has high threshold degree.
\newblock {\em {SIAM} J. Comput.}, 42(6):2329--2374, 2013.

\bibitem{sherstovrobust}
Alexander~A. Sherstov.
\newblock Making polynomials robust to noise.
\newblock {\em Theory of Computing}, 9:593--615, 2013.

\bibitem{sherstovhalfspaces2}
Alexander~A. Sherstov.
\newblock Optimal bounds for sign-representing the intersection of two
  halfspaces by polynomials.
\newblock {\em Combinatorica}, 33(1):73--96, 2013.

\bibitem{sherstov14}
Alexander~A. Sherstov.
\newblock Breaking the {M}insky-{P}apert barrier for constant-depth circuits.
\newblock In David~B. Shmoys, editor, {\em Symposium on Theory of Computing,
  {STOC} 2014, New York, NY, USA, May 31 - June 03, 2014}, pages 223--232.
  {ACM}, 2014.

\bibitem{shi}
Yaoyun Shi.
\newblock Approximating linear restrictions of boolean functions.
\newblock 2002.
\newblock Manuscript. Available online at:
  \texttt{web.eecs.umich.edu/~shiyy/mypapers/linear02-j.ps}.

\bibitem{shizhu}
Yaoyun Shi and Yufan Zhu.
\newblock Quantum communication complexity of block-composed functions.
\newblock {\em Quantum Information {\&} Computation}, 9(5):444--460, 2009.

\bibitem{spalek}
Robert Spalek.
\newblock A dual polynomial for {OR}.
\newblock {\em CoRR}, abs/0803.4516, 2008.

\bibitem{tal3}
Avishay Tal.
\newblock Shrinkage of de morgan formulae from quantum query complexity.
\newblock {\em Electronic Colloquium on Computational Complexity {(ECCC)}},
  21:48, 2014.

\bibitem{tal2}
Avishay Tal.
\newblock The bipartite formula complexity of inner-product is quadratic.
\newblock {\em Electronic Colloquium on Computational Complexity {(ECCC)}},
  23:181, 2016.

\bibitem{tal1}
Avishay Tal.
\newblock Computing requires larger formulas than approximating.
\newblock {\em Electronic Colloquium on Computational Complexity {(ECCC)}},
  23:179, 2016.

\bibitem{thaler}
Justin Thaler.
\newblock Lower bounds for the approximate degree of block-composed functions.
\newblock {\em Electronic Colloquium on Computational Complexity {(ECCC)}},
  21:150, 2014.
\newblock To appear in \emph{ICALP}, 2016.

\bibitem{difpriv1}
Justin Thaler, Jonathan Ullman, and Salil~P. Vadhan.
\newblock Faster algorithms for privately releasing marginals.
\newblock In Artur Czumaj, Kurt Mehlhorn, Andrew~M. Pitts, and Roger
  Wattenhofer, editors, {\em Automata, Languages, and Programming - 39th
  International Colloquium, {ICALP} 2012, Warwick, UK, July 9-13, 2012,
  Proceedings, Part {I}}, volume 7391 of {\em Lecture Notes in Computer
  Science}, pages 810--821. Springer, 2012.

\end{thebibliography}

\appendix

\section{A Refined Dual Witness for $\OR$} \label{app:or-dual}

Our goal is to prove the following equivalent formulation of Lemma~\ref{lem:or-sym-dual}.

\begin{lemma} \label{lem:or-sym-dual-ref}
Let $k \in \N$. There exists a constant $c_1 \in (0, 1) $ and a function $\omega : \{0, 1, \dots, k\} \to \R$ such that
\begin{equation}\omega(0) - \sum_{t = 1}^k \omega(t) \ge \frac{1}{3}\|\omega\|_1  \label{eqn:or-sym-corr-ref}\end{equation}
\begin{equation}  \text{For all univariate polynomials } q \colon \R \to \R\text{, } \deg p < c_1\sqrt{k} \implies \sum_{t=0}^k \omega(t) \cdot q(t) = 0 \label{eqn:or-sym-phd-ref}\end{equation}
\begin{equation} \omega(0) > 0 \label{eqn:or-sym-onesided-ref}\end{equation}
\begin{equation} \omega(t) \le \frac{5\|\omega\|_1}{(t+1)^2} \qquad \forall t = 0, 1, \dots, k\label{eqn:or-sym-decay-ref}\end{equation}
\end{lemma}
In the proof of Lemma~\ref{lem:or-sym-dual-ref}, we make use of the following combinatorial identity.

\begin{fact} \label{fact:combinatorial}
Let $k \in \N$, and let $p$ be a polynomial of degree less than $k$. Then
\[\sum_{t = 0}^k (-1)^t \binom{k}{t} p(t) = 0.\]
\end{fact}

\begin{proof}[Proof of Lemma~\ref{lemma:omega}]
Let $c = 25$ below. Let $m = \lfloor \sqrt{k/c} \rfloor$ and define the set
\[T = \{1, 2\} \cup \{ci^2 : 0 \le i \le m\}.\]
Note that $|T| = \Omega(\sqrt{k})$. Define the polynomial
\[\omega(t) = \frac{(-1)^{t+(k-m)}}{k!} \binom{k}{t} \prod_{r \in [k]_0 \setminus T} (t - r).\]
It is immediate from Fact~\ref{fact:combinatorial} that $\omega$ satisfies~\eqref{eqn:or-sym-phd-ref} for $c_1=1/\sqrt{c}$. By inspection, we have $\omega(0) > 0$ and $\omega(1) < 0$.

Expanding out the binomial coefficient reveals that
\[|\omega(t)| = \begin{cases}
 \prod\limits_{r \in T \setminus \{t\}} \frac{1}{|t - r|} & \text{ for } t \in T, \\
0 & \text{ otherwise.}
\end{cases}\]

For $t = 1$, we observe
\[\frac{|\omega(1)|}{\omega(0)} = \frac{2\prod_{i = 1}^m ci^2}{\prod_{i = 1}^m (ci^2 - 1)} = 2\prod_{i=1}^m \frac{i^2}{i^2 - 1/c} \ge 2.\]
On the other hand, for $t = 2$, we have
\begin{align} \label{eqn:omega-2}
\frac{|\omega(2)|}{\omega(0)} &= \frac{2\prod_{i = 1}^m ci^2}{2\prod_{i = 1}^m (ci^2 - 2)} \nonumber \\
&= \left(\prod_{i = 1}^m \frac{i^2 - 2/c}{i^2}\right)^{-1} \nonumber \\
&\le \left(1 - \sum_{i=1}^m \frac{2}{ci^2}\right)^{-1} \nonumber \\
&\le \left(1 - \frac{\pi^2}{3c}\right)^{-1} = \frac{3c}{3c - \pi^2}
\end{align}
where the first inequality follows from the fact that $\prod_{i = 1}^m (1 - a_i) \ge 1 - \sum_{i = 1}^m a_i$ for $a_i \in (0, 1)$.

For $t = cj^2$ with $j \ge 1$, we get
\begin{align*}
\frac{|\omega(t)|}{\omega(0)} &= \frac{2\prod_{i = 1}^m ci^2}{(cj^2 - 1)(cj^2-2)\prod_{i \in [m] \setminus \{j\}} |ci^2 - cj^2|} \\
&= \frac{2(m!)^2}{(c^2j^4 - 3cj^2 + 2)\prod_{i \in [m] \setminus \{j\}} (i+j)|i-j|} \\
&= \frac{2(m!)^2}{(c^2j^4 - 3cj^2 + 2)(m+j)!(m-j)!} \\
&\le \frac{2}{c^2j^4 - 3cj^2 + 2}
\end{align*}
where the last inequality follows because
\[\frac{(m!)^2}{(m+j)!(m-j)!} = \frac{m}{m+j} \cdot \frac{m-1}{m+j-1} \cdot \ldots \cdot \frac{m-j+1}{m+1}\]
is a product of factors that are each smaller than 1. Since
\[\frac{|\omega(cj^2)|}{\|\omega\|_1} \le \frac{|\omega(t)|}{\omega(0)} \le \frac{2}{c^2j^4 - 3cj^2 + 2} \le \frac{5}{(cj^2 + 1)^2}\]
for $c \ge 8$, this establishes~\eqref{eqn:or-sym-decay-ref}.

What remains is to perform the correlation calculation to establish~\eqref{eqn:or-sym-corr-ref}.
First, observe that the total contribution of $t > 2$ to $\|\omega\|_1/\omega(0)$ is at most
\begin{equation} \label{eqn:omega-tail}
\sum_{t > 2} \frac{|\omega(t)|}{\omega(0)} = \sum_{j=1}^{m} \frac{2}{c^2j^4 - 3cj^2 + 2} < \sum_{j=1}^\infty \frac{2}{cj^2} < \frac{\pi^2}{3c}.
\end{equation}

Next, we calculate
\begin{align}
\omega(0) - \sum_{t = 1}^k \omega(t) &\ge \omega(0) - \omega(1) - \left( \sum_{t = 2}^k |\omega(t)|\right) \nonumber \\
&\ge \omega(0) - \omega(1) - \left( \omega(2) +  \omega(0) \cdot \frac{\pi^2}{3c}\right) & \text{by } \eqref{eqn:omega-tail} \nonumber \\
&\ge - \omega(1) + \omega(0) \left(1 - \frac{3c}{3c-\pi^2} - \frac{\pi^2}{3c}\right) & \text{by } \eqref{eqn:omega-2} \nonumber \\
&\ge - \omega(1) - \frac{1}{3}\omega(0) & \text{by our choice of } c = 25. \label{eqn:fuckthis}
\end{align}
On the other hand,
\begin{align}
\|\omega\|_1 &\le \omega(0) - \omega(1) + \omega(2) + \omega(0) \cdot \frac{\pi^2}{3c} & \text{by } \eqref{eqn:omega-tail} \nonumber \\
&\le -\omega(1) + \omega(0) \left(1 + \frac{3c}{3c - \pi^2} + \frac{\pi^2}{3c} \right) &\text{by } \eqref{eqn:omega-2}\nonumber \\
&\le -\omega(1) + \frac{7}{3}\omega(0) & \text{since } c = 25. \label{eqn:fuckthat}
\end{align}
Combining~\eqref{eqn:fuckthis} and~\eqref{eqn:fuckthat}, and using the fact that $-\omega(1) \ge 2\omega(0)$ shows that
\[\frac{\omega(0) - \sum_{t = 1}^k \omega(t)}{\|\omega\|_1} \ge \frac{-\omega(1) - \frac13 \omega(0)}{-\omega(1) + \frac73 \omega(0)} \ge \frac{1}{3}.\]
This establishes~\eqref{eqn:or-sym-corr-ref}, completing the proof.

\end{proof}

\end{document}